\documentclass[12pt]{article} 
\usepackage[dvips]{graphicx,color}%
\usepackage{amsmath,amssymb}

\newtheorem{theorem}{Theorem}[section]

\newcommand{\qed}{\qquad$\blacksquare$\\}

\newtheorem{proposition}[theorem]{Proposition}
\newtheorem{lemma}[theorem]{Lemma}
\newtheorem{corollary}[theorem]{Corollary}
\newtheorem{definition}[theorem]{Definition}

\numberwithin{equation}{section}

\begin{document} 

\title{A Mathematical Aspect of A Tunnel-Junction \\ 
for Spintronic Qubit}

\author{Masao Hirokawa\thanks{Department of Mathematics, Okayama University, 
Okayama, 700-8530, Japan ({\tt hirokawa@math.okayama-u.ac.jp}).}
        \and Takuya Kosaka\thanks{Department of Mathematics, 
Okayama University, Okayama, 700-8530, Japan.}}

\date{\today}

\markright{M. Hirokawa \& T. Kosaka}
\pagestyle{myheadings}

\maketitle

\begin{abstract}
We consider the Dirac particle living in the $1$-dimensional 
configuration space with a junction for a spintronic qubit. 
We give concrete formulae explicitly showing the one-to-one 
correspondence between every self-adjoint extension of 
the minimal Dirac operator and the boundary condition of 
the wave functions of the Dirac particle. 
We then show that the boundary conditions are classified into 
two types: one of them is characterized by two parameters and 
the other is by three parameters. 
Then, we show that Benvegn\`{u} and D\c{a}browski's four-parameter 
family can actually be characterized by three parameters, 
concerned with the reflection, penetration, and phase factor.  
\end{abstract}

\pagestyle{myheadings}
\thispagestyle{plain}
\markboth{M. HIROKAWA AND T. KOSAKA}{Tunnel-Junction Formulae for Spintronic Qubit}

\section{Introduction} 

The current cutting-edge technology has developed 
seeking the realization of quantum information 
and quantum computation. 
For such realization, the mathematical modeling of the system of 
some quantum devices will play an important role. 
One of the candidates for qubit is electron spin, 
that is, the so-called \textit{electron-spin qubit} or 
the \textit{spintronic qubit}, 
from the point of the view of spintronics 
\cite{AFS02, BGMBPA06, Burkard06, LLPA02, LFA05, LDV98}. 
It is remarkable that the transportation of 
a single electron as a qubit has been demonstrated 
experimentally \cite{UniTokyo11, Cambridge11} 
as well as the spin-conserved transport tunneling through a junction 
was experimentally demonstrated in organic spintronics 
\cite{MIT07,Tohoku13}. 
In addition, the spin-flip and the phase-shift was demonstrated 
in the experiment for the spin state of an electron-hole 
pair in a semiconductor quantum dot \cite{GBPSSPFRS10}, 
while the spin alignment was studied 
in the process of the spin-transport in organic spintronics 
\cite{Eindhoven09}. 
The experimental development makes us hopeful that 
an integrated circuit for qubit and a quantum network 
are really demonstrated in future.

This paper deals with the Dirac particle living in the configuration space 
consisting of the two quantum wires and a junction for the spintronic qubit. 
Although we actually have to determine a concrete physical object for the junction, 
we regard the junction as a black box 
so that it has mathematical arbitrariness. 
We regard the wires as the union of the intervals, 
$(-\infty,-\Lambda)\cup(\Lambda , \infty)$, $\Lambda\ge 0$, 
for mathematical simplicity. 
Namely, the segment $[-\Lambda , \Lambda]$ with length $2\Lambda$ 
plays a role of the junction. 
Many physicists have investigated the individual boundary condition 
of the wave functions of the Dirac particle 
for the corresponding self-adjoint extension of the Dirac operator in the case 
of the point interaction (i.e., $\Lambda=0$) 
\cite{AGH-KH,ADeV00,BD94,D-AM89,Seba89}.  
Meanwhile, the Dirac operator consists of the combination 
of the Dirac matrices and the momentum operator of electron. 
The boundary conditions of the self-adjoint extensions 
of the momentum operators 
have been studied by mathematicians \cite{Hir00,PT13}. 
We note that there is a general theory in mathematics, 
called boundary triple, 
to handle the boundary condition, and the theory has still been developed 
\cite{BMN08,BNRRW08,BGP08,BMN02,CMP13}. 
In Refs.\cite{FHNS10, SH11}  the appearance of a phase factor 
was proved for the Schr\"{o}dinger particle 
under the same mathematical set-up as ours. 
For our Dirac particle, in the case where $\Lambda=0$, 
Benvegn\`{u} and D\c{a}browski showed that 
a phase factor appears in their four-parameter family 
(see Eq.(15) of Ref.\cite{BD94}). 
In addition, they showed how the boundary condition affects the spin. 
The description of boundary condition by the individual one-parameter families have been studied 
in Refs.\cite{AGH-KH,D-AM89,Hughes97,Seba89} (also see Eqs.(17) and (18) of Ref.\cite{BD94}). 
We will go ahead and make an in-depth research for the Dirac particle. 
Thus, we employ the minimal Dirac operator for the Hamiltonian. 
In this paper, we follow the machinery in Refs.\cite{BD94,FHNS10,Hir00,PT13,SH11} 
based on the von Neumann's theory \cite{RS2,weidmann}, 
in which all the self-adjoint extensions of 
our minimal Dirac operator are parameterized 
by $U\in U(2)$, where $U(2)$ is unitary group of degree $2$. 

We prove that all the boundary conditions of wave functions of our Dirac particle 
are completely classified into the two types (See Corollary \ref{cor:theo-2}). 
One of them is the type that states 
the wave functions do not pass through the junction, 
and described by two parameters, $\gamma_{L}, \gamma_{R}\in\mathbb{C}$ 
with $|\gamma_{L}|=|\gamma_{R}|=1$, 
concerned with the reflections at $-\Lambda$ and at $+\Lambda$, respectively. 
The other is the type that states 
the wave functions do pass through the junction, 
and described by Benvegn\`{u} and D\c{a}browski's four-parameter family 
(see Theorem \ref{theo:1}, Theorem \ref{theo:2}, and Proposition \ref{prop:06}). 
We give our concrete formulae showing the one-to-one correspondence 
between the boundary conditions of wave functions of our Dirac particle 
and the self-adjoint extensions of the minimal Dirac operator 
(see  Theorem \ref{theo:2}, 
and Propositions \ref{prop:03'}, \ref{prop:inverse}, \ref{prop:06}). 
Our formulae states that Benvegn\`{u} and D\c{a}browski's four-parameter family 
can be characterized by three parameters, $\gamma_{1}, \gamma_{2}, \gamma_{3} 
\in \mathbb{C}$ with $|\gamma_{1}|^{2}+|\gamma_{2}|^{2}=|\gamma_{3}|=1$, 
concerned with the reflection, penetration, and phase factor, respectively 
(see Corollary \ref{cor:BD->3}). 
To obtain our formulae, we invent our representation of $U(2)$ in Proposition \ref{prop:03'}. 
We then realize that all the boundary conditions are independent of $2\Lambda$, 
the length of the junction, which is different from the case for the Schr\"{o}dinger 
operator \cite{FHNS10,HK13-S,SH11}. 
Through our mathematical toy model we propose a mathematically fundamental idea 
for a tunnel-junction device 
in the light of the industry of quantum engineering.

\section{Mathematical Notations and Notions} 

We always assume that every Hilbert space $\mathcal{H}$ 
that we handle in this paper is separable. 
We denote by $\langle\,\,\,|\,\,\,\rangle_{\mathcal{H}}$ 
the inner product of the Hilbert space $\mathcal{H}$, 
where we suppose that the right hand side of the inner product 
$\langle\,\,\,|\,\,\,\rangle_{\mathcal{H}}$ is linear. 
We here prepare some notations and notions using in 
operator theory.

Let $\mathcal{L}(\mathcal{H})$ 
denote the set of all (linear) operators acting in $\mathcal{H}$. 
We always omit the word of `linear' from the notion of 
linear operator because 
we consider only linear operators in this paper. 
So, when we write, $A \in \mathcal{L}(\mathcal{H})$, 
it means that $A$ is an operator acting in the Hilbert space $\mathcal{H}$. 
That is, there is a subspace $D(A)\subset \mathcal{H}$ 
so that $A$ is a linear map from $D(A)$ to $\mathcal{H}$. 
We call $D(A)$ the \textit{domain} of the operator $A$. 
We say that an operator $B\in \mathcal{L}(\mathcal{H})$ 
is an \textit{extension} of the operator $A\in \mathcal{L}(\mathcal{H})$, 
provided that 
$D(A)\subset D(B)$ and $A\psi=B\psi$ for every $\psi\in D(A)$.  
We denote it as $A\subset B$ or 
$B\supset A$. 
The operator equation is defined in the following: 
an operator $A\in \mathcal{L}(\mathcal{H})$ 
is \textit{equal} to an operator $B\in \mathcal{L}(\mathcal{H})$ 
if and only if $D(A)= D(B)$ and 
$A\psi=B\psi$ for every $\psi\in D(A)=D(B)$.  
In particular, the operator equation, $A=B$, is equivalent to 
the conditions, $A\subset B$ and $A\supset B$. 
For every operator $A\in \mathcal{L}(\mathcal{H})$ 
and subspace $S\subset\mathcal{H}$ with $S\subset D(A)$, 
the operator $A\lceil{S}\in\mathcal{L}(\mathcal{H})$ is defined by 
$D(A\lceil{S}):=S$ and $A\lceil{S}\psi:=A\psi$ 
for every $\psi\in S$. 
We call $A\lceil{S}$ the \textit{restriction} of 
the operator $A$ on the subspace $S$.
In particular, the equation, 
$B\lceil{D(A)}=A$, holds for operators 
$A, B\in \mathcal{L}(\mathcal{H})$ with 
$A\subset B$. 
Let $I_{\mathcal{H}}$ or 
$I$ denote 
the identity operator on $\mathcal{H}$, i.e., 
$I_{\mathcal{H}}\psi\equiv I\psi:=
\psi$ for every $\psi\in\mathcal{H}$. 

An operator $A\in \mathcal{L}(\mathcal{H})$ is said 
to be {\it closed} when the graph $\mathcal{G}(A)$ 
of the operator $A$ is closed in $\mathcal{H}\times \mathcal{H}$, 
where $\mathcal{G}(A):=\left\{ (\psi\, ,\, A\psi) 
\in \mathcal{H}\times\mathcal{H}\, |\, 
\psi\in D(A)\right\}$. 
In other words, if $D(A)\ni\psi_{n}\to\psi\in\mathcal{H}$ 
and $A\psi_{n}\to\phi$ in $\mathcal{H}$ as $n\to\infty$, 
then $\psi\in D(A)$ and 
$A\psi=\phi$. 
We say that an operator $A\in\mathcal{L}(\mathcal{H})$ 
is \textit{densely defined} 
when its domain $D(A)$ is dense in $\mathcal{H}$.

Let an operator $A\in \mathcal{L}(\mathcal{H})$ be densely defined. 
We define a subspace $D_{A^{*}}$ of $\mathcal{H}$ by 
$$
D_{A^{*}} := 
\biggl\{ 
\psi\in\mathcal{H}\, \bigg|\, 
\textrm{there is}\,\,\, \phi_{\psi}\in\mathcal{H}\,\,\, 
\textrm{so that}\,\,\, 
\langle\psi|A\varphi\rangle_{\mathcal{H}}=
\langle\phi_{\psi}|\varphi\rangle_{\mathcal{H}}\,\,\, 
\textrm{for every}\,\,\, 
\varphi\in D(A)\biggr\}. 
$$ 
We note $\phi_{\psi}$ is uniquely determined then.  
The \textit{adjoint operator} $A^{*}$ of $A$ is defined by 
$D(A^{*}) := D_{A^{*}}$ and $A^{*}\psi := \phi_{\psi}$. 
It is well known that $A^{*}$ is closed and 
$\langle\psi|A\varphi\rangle_{\mathcal{H}}=
\langle A^{*}\psi|\varphi\rangle_{\mathcal{H}}$. 
A densely defined operator $A\in \mathcal{L}(\mathcal{H})$ 
is \textit{symmetric} when the condition, $A\subset A^{*}$, holds. 
We say that a densely defined operator $A\in \mathcal{L}(\mathcal{H})$ 
is \textit{self-adjoint} if and only if the operator equation, 
$A=A^{*}$, holds. 
We emphasize that self-adjointness requires $D(A)=D(A^{*})$, 
while symmetry requires $D(A)\subset D(A^{*})$ only.

Let $A_{0} \in \mathcal{L}(\mathcal{H})$ be closed symmetric. 
An operator $A\in \mathcal{L}(\mathcal{H})$ is 
\textit{self-adjoint extension} of $A_{0}$, 
provided that the condition, $A_{0}\subset A$, holds and 
the operator $A$ is self-adjoint. 
For every closed symmetric operator $A_{0}\in\mathcal{L}(\mathcal{H})$, 
we respectively define \textit{deficiency subspaces} $\mathcal{K}_{+}(A_{0})$ 
and $\mathcal{K}_{-}(A_{0})$ by $\mathcal{K}_{\pm}(A_{0})
:=\mathrm{ker}(\pm i-A_{0}^{*})$, and moreover, 
\textit{deficiency indices} $n_{+}(A_{0})$ and $n_{-}(A_{0})$ by 
$n_{\pm}(A_{0}):=\mathrm{dim}\, \mathcal{K}_{\pm}(A_{0})$. 
Namely, $\mathcal{K}_{\pm}(A_{0})$ are the eigenspaces 
of $A_{0}^{*}$, respectively, corresponding to 
the eigenvalues $\pm i$, 
and $n_{\pm}(A_{0})$ are their individual dimensions.

A \textit{unitary operator} $U$ from a Hilbert space $\mathcal{H}_{1}$ 
to a Hilbert space $\mathcal{H}_{2}$ is defined as follows: 
$U$ is a surjective linear map from the Hilbert space $\mathcal{H}_{1}$ 
to the Hilbert space $\mathcal{H}_{2}$, and it satisfies 
the relation, 
$\langle U\psi\, |\, U\varphi\rangle_{\mathcal{H}_{2}}
=\langle \psi\, |\, \varphi\rangle_{\mathcal{H}_{1}}$ 
for every $\psi, \varphi \in\mathcal{H}_{1}$.  

Our argument developed in this paper is based on the following 
proposition: 

\begin{proposition}
\label{prop:von-Neumann}
(von Neumann \cite{RS2,weidmann}): 
Let $A_{0}\in\mathcal{L}(\mathcal{H})$ be closed symmetric. 
\begin{enumerate}
\item[i)] If $n_{+}(A_{0}) = n_{-}(A_{0})$, then 
$A_0$ has self-adjoint extensions. 
\item[ii)] There is a one-to-one correspondence 
between self-adjoint extensions $A$ of $A_{0}$ 
and unitary operators $U : \mathcal{K}_{+}(A_0) 
\to \mathcal{K}_{-}(A_0)$ so that the correspondence is 
given by the following: 
For every unitary operator $U : 
\mathcal{K}_{+}(A_{0}) \to \mathcal{K}_{-}(A_{0})$, 
the corresponding self-adjoint extension $A_U$ 
is defined by 
$$
\left\{ \begin{array}{l}
D(A_U):=\left\{\psi = \psi_{0} + \psi^{+} + U\psi^{+}\, 
\big|\, \psi_{0} \in D(A_0),\, \psi^{+} \in \mathcal{K}_{+}(A_0)
\right\}, \\
A_{U}:=A_{0}^{*} 
\lceil{D(A_U)}, 
\end{array}\right.
$$
and then its operation is 
$A_{U}(\psi_{0} + \psi^{+} + U\psi^{+}) 
= A_{0}\psi_{0} + i\psi^{+} -iU\psi^{+}$. 
Conversely, for every self-adjoint extension $A$ 
of $A_{0}$, there is the corresponding unitary operator 
$U: \mathcal{K}_{+}(A_{0}) 
\to \mathcal{K}_{-}(A_{0})$ 
so that 
$A=A_{U}$. 
\end{enumerate}
\end{proposition} 

Let us now suppose that $n_{+}(A_{0})=n_{-}(A_{0})=n \in \mathbb{N}$. 
Fix individual complete orthonormal systems, 
$\{ \psi_{j}^{\pm} \}_{j=1}^{n}$, of the deficiency subspaces 
$\mathcal{K}_{\pm}(A_{0})$. 
We identify unitary operators from 
$\mathcal{K}_{+}(A_{0})$ to $\mathcal{K}_{-}(A_{0})$ with 
$n\times n$ unitary matrices making the correspondence by 
$U : \psi_{j}^{+} \longmapsto \sum_{k=1}^{n} u_{jk} \psi_{k}^{-}$, 
$j = 1, \cdots , n$. 
So, we often identify the unitary operator $U$ 
with the unitary matrix $(u_{jk})_{jk}$, 
and write $U\in  U(n)$ in the case $n_{\pm}(A_{0})<\infty$, 
where $U(n)$ denotes the unitary group of degree $n$. 
We say $U$ is \textit{diagonal} if 
$u_{jk} = 0$ with $j\ne k$. 
Otherwise, we say $U$ is \textit{non-diagonal}.

We here introduce some notations concerning 
function spaces that we use in this paper. 
Let $\Omega$ be an open set of the $1$-dimensional Euclidean space 
$\mathbb{R}$, i.e., $\Omega\subset \mathbb{R}$. 
We respectively define function spaces, 
$L^{2}(\Omega)$, $AC^{1}\bigl(\overline{\Omega}\big)$, 
and $AC_{0}^{1}\bigl(\overline{\Omega}\big)$ as follows:
$$
L^{2}(\Omega):=\left\{ f : \Omega\to\mathbb{C}\,\,\, 
\textrm{is the Lebesgue measureble}\,\,\,\bigg|\, 
\int_{\Omega}|f(x)|^{2}dx<\infty\right\},  
$$ 
where the integral is the Lebesgue integral. 
$$
AC^{1}\bigl(\overline{\Omega}\big) := 
\biggl\{ f \in L^{2}(\Omega)\, \bigg| \, 
\textrm{$f$ is absolutely continuous on 
$\overline{\Omega}$},\,\,\, 
\textrm{and}\,\,\, f{\,}' \in L^{2}(\Omega) \biggr\}. 
$$ 
Here $\overline{\Omega}$ denotes the closure of 
the set $\Omega$ in $\mathbb{R}$. 
We note that the differential $f{\,}'(x)$ of the 
absolutely continuous function $f(x)$ exists 
for almost every $x\in\overline{\Omega}$. 
$$
AC_{0}^{1}\bigl(\overline{\Omega}\bigr) := 
\left\{ f \in AC^{1}\bigl(\overline{\Omega}\bigr)\, 
\bigg|\, 
f=0\,\,\, 
\textrm{on}\,\,\, \partial \Omega
\right\}.
$$
Here  
$\partial\Omega$ is the boundary of the set $\Omega$. 
In the case where $+\infty$ (resp. $-\infty$) is in the set $\Omega$, 
the boundary condition in the function space 
$AC_{0}^{1}\bigl(\overline{\Omega}\bigr) $ means 
that $\lim_{x\to+\infty}f(x)=\lim_{x\to+\infty}f{\,}'(x)=0$ 
(resp. $\lim_{x\to-\infty}f(x)=\lim_{x\to-\infty}f{\,}'(x)=0$).

\section{One-Dimensional Dirac Operators} 

In this section we define some $1$-dimensional Dirac operators. 
First up, we define the configuration space in which 
the Dirac particle lives. 
For every $\Lambda \geq 0$, we set two intervals 
$\Omega_{\Lambda, L}$ and $\Omega_{\Lambda, R}$ by 
$\Omega_{\Lambda, L}:=(-\infty, -\Lambda)$ and 
$\Omega_{\Lambda, R}:=(+\Lambda, +\infty)$. 
We often call $\Omega_{\Lambda, L}$ and $\Omega_{\Lambda, R}$ 
the \textit{left island} and the \textit{right island}, 
respectively. 
We set our configuration space $\Omega_{\Lambda}$ by 
$\Omega_{\Lambda}:=\Omega_{\Lambda, L} \cup \Omega_{\Lambda, R}$. 
We define two spaces of functions 
on our configuration space $\Omega_{\Lambda}$ as: 
$$
\mathcal{AC}(\overline{\Omega_{\Lambda}}) 
:= \mathbb{C}^{2}\hat{\otimes} 
AC^{1}(\overline{\Omega_{\Lambda}})\,\,\, 
\textrm{and}\,\,\, 
\mathcal{AC}_{0}(\overline{\Omega_{\Lambda}}) 
:= \mathbb{C}^{2} \hat{\otimes} 
AC_{0}^{1}(\overline{\Omega_{\Lambda}}),
$$
where $\hat{\otimes}$ denotes the algebraic 
tensor product. 
The function spaces $\mathcal{AC}(\overline{\Omega_{\Lambda}})$ 
and $\mathcal{AC}_{0}(\overline{\Omega_{\Lambda}})$ 
are dense in $\mathbb{C}^{2} \otimes L^{2}(\Omega_{\Lambda})$, 
where $\otimes$ denotes the tensor product of 
Hilbert spaces. 
We, however, note the following. 
The Sobolev spaces respectively corresponding to 
$\mathcal{AC}(\overline{\Omega_{\Lambda}})$ 
and $\mathcal{AC}_{0}(\overline{\Omega_{\Lambda}})$ 
have their own Sobolev-space structures 
different from each other 
because of the junction $[-\Lambda,+\Lambda]$ for 
$\Lambda\ge 0$, which implies existence of 
uncountably many self-adjoint extensions of 
the minimal Dirac operator defined below.

For the quantization of a relativistic particle on 
the $1$-dimensional configuration space $\Omega_{\Lambda}$, 
we seek a representation of 
`energy $=\alpha\otimes p+\beta\otimes mI$' 
with matrices $\alpha$ and $\beta$ satisfying 
$\alpha^{2}=\beta^{2}=I$ and
$\alpha\beta+\beta\alpha=0$ 
for the probability interpretation of 
the wave function of electron. 
Here $p$ is the momentum operator and 
$m$ is the mass of electron. 
Then, we have candidates of the representation: 
$\alpha=\sigma_{x}$ or $\sigma_{y}$, 
and $\beta=\sigma_{z}$, 
where $\sigma_{x}, \sigma_{y}, \sigma_{z}$ are 
the Pauli (spin) matrices:  
$$
\sigma_{x}:=
\left(
\begin{array}{cc}
0 & 1 \\ 
1 & 0
\end{array}
\right),\,\,\, 
\sigma_{y}:=
\left(
\begin{array}{cc}
0 & -i \\ 
i & 0
\end{array}
\right),\,\,\, 
\text{and}\,\,\,  
\sigma_{z}:=
\left(
\begin{array}{cc}
1 & 0 \\ 
0 & -1
\end{array}
\right).
$$ 
We employ $\sigma_{x}$ as $\alpha$ throughout this paper.

\begin{definition}
{\rm (Minimal Dirac Operator): 
Let $m \ge 0$ be the mass of electron. 
The $1$-dimensional Dirac operator $H_{0} 
\in \mathcal{L}(\mathbb{C}^{2} \otimes L^{2}(\Omega_\Lambda))$ 
is defined by
$$
\left\{ \begin{array}{l}
D(H_{0}):=
\mathcal{AC}_{0}(\overline{\Omega_{\Lambda}}), \\ 
H_{0}:= 
\sigma_{x} \otimes p 
+ m \sigma_{z} \otimes I_{L^{2}(\Omega_{\Lambda})},
\end{array}\right.
$$
where the momentum operator $p$ is given by 
$p:=\, 
-i\frac{d}{dx}$.  
We call the operator $H_{0}$ the \textit{minimal Dirac operator}.
}  
\end{definition}

The following proposition comes from the well-known facts 
on the momentum operator $p$ and its adjoint operator $p^{*}$: 
\begin{proposition} 
\label{prop:adjoint-operator}
The minimal Dirac operator $H_0$ is closed symmetric. 
Its adjoint operator $H_0^*$ is given by 
$$
\left\{ \begin{array}{l}
D(H_{0}^{*}) = \mathcal{AC}(\overline{\Omega_{\Lambda}}), \\ 
H_{0}^{*}
=\sigma_{x}\otimes p^{*} + m\sigma_{z}\otimes I_{L^{2}(\Omega_{\Lambda})}, 
\end{array}\right.
$$
where
$p^{*}=\, -i\frac{d}{dx}$. 
\end{proposition}

Since the operation of the adjoint operator $H_{0}^{*}$ 
is the same as that of the minimal Dirac operator $H_{0}$ 
though their domains are different from each other, 
we make the following definition: 

\begin{definition}
{\rm (Maximal Dirac Operator): 
We call the adjoint operator $H_{0}^{*}$ 
the \textit{maximal Dirac operator}. 
}
\end{definition}

Similarly, since the restriction of the adjoint operator $H_{0}^{*}$ 
on every subspace $\mathcal{D}$ with the condition, 
$D(H_{0})\subset\mathcal{D}\subset D(H_{0}^{*})$, 
has the same operation as that of the minimal Dirac operator $H_{0}$, 
we name them in the following: 

\begin{definition}
{\rm For every subspace $\mathcal{D}$ with the condition, 
$D(H_{0})\subset\mathcal{D}\subset D(H_{0}^{*})$, 
we call the restriction $H_{0}^{*}\lceil\mathcal{D}$ 
the \textit{Dirac operator} in this paper. 
}
\end{definition}

\section{Main Results}

Set constants $\mu$ and $N$ by $\mu:=(1+im)/\sqrt{1+m^{2}}$ 
and 
$N:=(1+m^{2})^{\frac{1}{4}}e^{-\sqrt{1+m^{2}}\Lambda}$, respectively. 
We define functions $\psi_{L}^{+}(x)$ and $\psi_{R}^{+}(x)$ by
\begin{equation}
\left\{ \begin{array}{l}
\psi_{L}^{+}(x)
\equiv 
\left(
\begin{array}{c}
\psi_{L\uparrow}^{+}(x) \\ 
\psi_{L\downarrow}^{+}(x)
\end{array}
\right)
:= 
N{\displaystyle 
\left(
\begin{array}{c}
1 \\ 
-\mu
\end{array}
\right)
\otimes\chi_{L}(x)e^{\sqrt{1+m^{2}}\, x}
}, 
\vspace*{3mm} \\ 
\psi_{R}^{+}(x)
\equiv 
\left(
\begin{array}{c}
\psi_{R\uparrow}^{+}(x) \\ 
\psi_{R\downarrow}^{+}(x)
\end{array}
\right)
:= 
N{\displaystyle 
\left(
\begin{array}{c}
1 \\ 
\mu
\end{array}
\right)
\otimes \chi_{R}(x)e^{-\sqrt{1+m^{2}}\, x}
},  
\end{array}\right.
\label{eq:+eigenfunctions}
\end{equation}
and functions $\psi_{L}^{-}(x)$ and $\psi_{R}^{-}(x)$ by
\begin{equation}
\left\{ \begin{array}{l}
\psi_{L}^{-}(x)
\equiv 
\left(
\begin{array}{c}
\psi_{L\uparrow}^{-}(x) \\ 
\psi_{L\downarrow}^{-}(x)
\end{array}
\right)
:= 
N{\displaystyle 
\left(
\begin{array}{c}
1 \\ 
\mu^{*}
\end{array}
\right)
\otimes\chi_{L}(x)e^{\sqrt{1+m^{2}}\, x}
}, 
\vspace*{3mm} \\ 
\psi_{R}^{-}(x)
\equiv 
\left(
\begin{array}{c}
\psi_{R\uparrow}^{-}(x) \\ 
\psi_{R\downarrow}^{-}(x)
\end{array}
\right)
:= 
N{\displaystyle 
\left(
\begin{array}{c}
1 \\ 
-\mu^{*}
\end{array}
\right)
\otimes\chi_{R}(x)e^{-\sqrt{1+ m^{2}}\, x}
}.  
\end{array}\right.
\label{eq:-eigenfunctions}
\end{equation}
Here $\chi_{L}$ and $\chi_{R}$ are respectively 
the characteristic functions on the closure 
$\overline{\Omega_{\Lambda, L}}$ of the left island 
and the closure $\overline{\Omega_{\Lambda, R}}$ 
of the right island.

As proved in \S\ref{subsec:proof-deficiency-indices}, 
we can compute the deficiency indices 
of the minimal Dirac operator $H_{0}$ in the following: 
\begin{proposition}
\label{prop:deficiency-indices} 
The deficiency indices 
of the minimal Dirac operator are $n_{+}(H_{0})=n_{-}(H_{0})=2$ 
and therefore the minimal Dirac operator $H_{0}$ has 
self-adjoint extensions. 
Then, the deficiency subspaces are given by 
$$\mathcal{K}_{+}(H_{0})=
\left\{ 
c_{L}\psi_{L}^{+}+c_{R}\psi_{R}^{+}\, |\, 
c_{L}, c_{R}\in\mathbb{C}\right\}\,\,\, 
\textrm{and}\,\,\, 
\mathcal{K}_{-}(H_{0})=
\left\{ 
c_{L}\psi_{L}^{-}+c_{R}\psi_{R}^{-}\, |\, 
c_{L}, c_{R}\in\mathbb{C}\right\}.
$$
\end{proposition}

By Propositions \ref{prop:von-Neumann} and \ref{prop:deficiency-indices} 
we can represent every self-adjoint extension of 
the minimal Dirac operator by an element of $U(2)$. 
We denote by $H_{U}$ the self-adjoint extension 
of the minimal Dirac operator $H_{0}$ 
corresponding to $U\in U(2)$.  
Then, Proposition \ref{prop:von-Neumann} says that 
the domain $D(H_{U})$ of the self-adjoint extension $H_{U}$ is 
$$
D(H_{U})=
\left\{\psi = \psi_{0} + \psi^{+} + U\psi^{+}\, 
\big|\, \psi_{0} \in D(H_{0}),\, \psi^{+} \in \mathcal{K}_{+}(H_{0})
\right\}. 
$$
Since $\psi_{0}(\pm\Lambda)=0$ 
for $\psi_{0}\in D(H_{0})
=\mathcal{AC}_{0}(\overline{\Omega_{\Lambda}})$, 
the unitary operator $U\in U(2)$ includes the information 
about how the electron reflects at the boundary and 
how it passes through the junction. 
So, our problem is to derive the boundary condition form 
the information that $U$ has. 
We determine the entries $u_{k\ell}$ 
of $U=(u_{k\ell})_{k,\ell=1,2}$ as follows: 
$$
\left\{ \begin{array}{l}
U\psi_{L}^{+}=u_{11}\psi_{L}^{-}+u_{12}\psi_{R}^{-}, \\ 
U\psi_{R}^{+}=u_{21}\psi_{L}^{-}+u_{22}\psi_{R}^{-}.
\end{array}\right.
$$
The von Neumann theory says that the unitary operator $U$ 
maps the eigenfunction $\psi_{L}^{+}$ living in the left island 
(resp. $\psi_{R}^{+}$ living in the right island) to 
the eigenfunction $\psi_{L}^{-}$ (resp. $\psi_{R}^{-}$) 
staying in the same island with the probability $|u_{11}|^{2}$ 
(resp. $|u_{22}|^{2}$) and the eigenfunction $\psi_{R}^{-}$ 
(resp. $\psi_{L}^{-}$) coming from the opposite island 
with the probability $|u_{12}|^{2}$ (resp. $|u_{21}|^{2}$).

We denote by $\overline{\mathbb{R}}$ 
the set of all extended real numbers: 
$\overline{\mathbb{R}}:=\mathbb{R}\cup\{+\infty\}$. 
For every two parameters $\rho= (\rho_{+}, \rho_{-})\in 
\overline{\mathbb{R}}^{2}$, 
we introduce a boundary condition by 
\begin{equation}
\left\{ \begin{array}{l}
i\rho_{+}\psi_{\uparrow}(+\Lambda)
=\psi_{\downarrow}(+\Lambda)\quad 
\textrm{if $\rho_{+} \in \mathbb{R}$}, \\
\psi_{\uparrow}(+\Lambda)=0\quad 
\textrm{if $\rho_{+}=+\infty$}, \\
i\rho_-\psi_{\uparrow}(-\Lambda)
=\psi_{\downarrow}(-\Lambda)\quad 
\textrm{if $\rho_{-} \in \mathbb{R}$}, \\
\psi_{\uparrow}(-\Lambda)=0\quad 
\textrm{if $\rho_{-}=+\infty$}.
\end{array}\right.  
\label{BC-rho}
\end{equation}

We introduce a class of four parameters 
$\alpha=(\alpha_{1}, \alpha_{2}, \alpha_{3}, \alpha_{4}) 
\in \mathbb{C}^{4}$ in the following: 
\begin{equation}
\left\{ \begin{array}{l}
\Re(\alpha_{1}\alpha_{2}^{*})
=\Re(\alpha_{1}\alpha_{3}^{*})=0, \\ 
\Re(\alpha_{2}\alpha_{4}^{*})
=\Re(\alpha_{3}\alpha_{4}^{*})=0, \\ 
\alpha_{1}\alpha_{4}^{*}+\alpha_{2}\alpha_{3}^{*}
=\alpha_{1}\alpha_{4}^{*}+\alpha_{2}^{*}\alpha_{3}=1.
\end{array}\right. 
\label{Cl-alpha}
\end{equation}
For every $\alpha \in\mathbb{C}^{4}$ in 
the class (\ref{Cl-alpha}), 
we define the \textit{boundary matrix} 
$B_{\alpha} \in M_{2}(\mathbb{C})$ by 
$$
B_{\alpha}:=
\left(
\begin{array}{cc}
\alpha_{1} & \alpha_{2} \\ 
\alpha_{3} & \alpha_{4} 
\end{array}
\right).
$$
We note that our vector $\alpha$ is equivalent 
to the Benvegn\`{u} and D\c{a}browski's four-parameter family 
given in Eq.(15) of Ref.\cite{BD94} 
(see Proposition \ref{prop:06} below).

For every wave function $\psi\in
\mathbb{C}^{2}\otimes L^{2}(\Omega_{\Lambda})$, 
we respectively set the wave function $\psi_{\uparrow}$ with 
\textit{up-spin} and 
the wave function $\psi_{\downarrow}$ with 
\textit{down-spin} by
$$
\psi_{\uparrow}:=
\left(
\begin{array}{c}
1 \\ 
0
\end{array}
\right)
\otimes \psi\qquad 
\textrm{and}\qquad 
\psi_{\downarrow}:=
\left(
\begin{array}{c}
0 \\ 
1
\end{array}
\right)
\otimes \psi.  
$$
Because of the unitarily equivalence, 
$\mathbb{C}^{2} \otimes L^{2}(\Omega_{\Lambda})\cong 
L^{2}(\Omega_{\Lambda}) \oplus L^{2}(\Omega_{\Lambda})$, 
we often identify $\mathbb{C}^{2} \otimes L^{2}(\Omega_{\Lambda})$ 
with $L^{2}(\Omega_{\Lambda}) \oplus L^{2}(\Omega_{\Lambda})$, 
and then, we represent $\psi\in \mathbb{C}^{2} \otimes L^{2}(\Omega_{\Lambda})$ 
as
$$
\psi =
\left(
\begin{array}{c}
\psi_{\uparrow} \\ 
\psi_{\downarrow}
\end{array}
\right)
\equiv 
{{}^{t}}(\psi_{\uparrow} , \psi_{\downarrow})
\in L^{2}(\Omega_{\Lambda}) \oplus L^{2}(\Omega_{\Lambda}).
$$ 
Conforming with this representation, 
we often use the representation: 
$$
\psi(\pm\Lambda)=
\left(
\begin{array}{c}
\psi_{\uparrow}(\pm\Lambda) \\ 
\psi_{\downarrow}(\pm\Lambda)  
\end{array}
\right) 
\equiv 
{{}^{t}}(\psi_{\uparrow}(\pm\Lambda) , \psi_{\downarrow}(\pm\Lambda)) 
\in\mathbb{C}^{2}.
$$
Then, the boundary matrix $B_{\alpha}$ gives a boundary condition: 
\begin{equation}
\left(
\begin{array}{c}
\psi_{\uparrow}(+\Lambda) \\ 
\psi_{\downarrow}(+\Lambda) 
\end{array}
\right)
\equiv 
\psi(+\Lambda) 
= 
B_{\alpha}\psi(-\Lambda)
\equiv 
\left(
\begin{array}{cc}
\alpha_{1} & \alpha_{2} \\ 
\alpha_{3} & \alpha_{4} \\ 
\end{array}
\right)
\left(
\begin{array}{c}
\psi_{\uparrow}(-\Lambda) \\ 
\psi_{\downarrow}(-\Lambda) 
\end{array}
\right). 
\label{BC-alpha}
\end{equation}

\qquad 

As proven in \S\ref{proof-theo-1}, 
the Dirac operators with the following two types of boundary conditions 
are self-adjoint extensions of the minimal Dirac operator:  
\begin{theorem}
\label{theo:1}
\begin{enumerate}
\item[i)] Give a subspace $D(H_{\rho})$ by 
$$
D(H_{\rho}):=
\left\{\psi \in D(H_{0}^{*}) \,\bigg|\, 
\psi\,\,\, \textrm{satisfies the boundary condition 
(\ref{BC-rho})}\right\} 
$$  
for every $\rho\in\overline{\mathbb{R}}^{2}$. 
Then, the Dirac operator $H_{\rho}$ defined 
as the restriction of the maximal Dirac operator $H_{0}^{*}$ 
on $D(H_{\rho})$, i.e., $H_{\rho}:=H_{0}^{*}\lceil{D(H_{\rho})}$, 
is a self-adjoint extension of the minimal Dirac operator $H_{0}$.
\item[ii)] Give a subspace $D(H_{\alpha})$ by 
$$
D(H_{\alpha}):=
\left\{\psi \in D(H_{0}^{*}) \,\bigg|\, 
\psi\,\,\, \textrm{satisfies the boundary condition 
(\ref{BC-alpha})}\right\} 
$$  
for every vector $\alpha$ in the class (\ref{Cl-alpha}). 
Then, the Dirac operator $H_{\alpha}$ defined 
as the restriction of the maximal Dirac operator $H_{0}^{*}$ 
on $D(H_{\alpha})$, i.e., $H_{\alpha}:=H_{0}^{*}\lceil{D(H_{\alpha})}$, 
is a self-adjoint extension of the minimal Dirac operator $H_{0}$. 
\end{enumerate}
\end{theorem}

Before stating our second theorem, 
we introduce a device for the representation 
of $U(2)$. 
In general, the so-called homomorphism theorem tells us that 
$U(n)/SU(n)\cong U(1)$ for each $n\in\mathbb{N}$, 
where $SU(n)$ is the special unitary group of degree $n$. 
In this paper we seek another representation of 
the unitary group $U(2)$ making good use of the degree, 
$n=2$, that we handle now. 
The following proposition will be 
proved in \S\ref{subsec:proof-prop-03'}: 
\begin{proposition}
\label{prop:03'} 
The unitary group $U(2)$ has the following representation: 
$$
U(2)=U(1)\, S\mathbb{H}
=\left\{\gamma_{3}
\left(
\begin{array}{cc}
\gamma_{1} & - \gamma_{2}^{*} \\ 
\gamma_{2} & \gamma_{1}^{*} 
\end{array}
\right)\, \bigg|\, 
\gamma_{1}, \gamma_{2}, \gamma_{3} \in \mathbb{C},\,\,\, 
|\gamma_{1}|^{2}+|\gamma_{2}|^{2}=|\gamma_{3}|= 1
\right\}. 
$$
Here $S\mathbb{H}$ is defined by  
$$
S\mathbb{H}:=
\left\{ 
A\in \mathbb{H}\, |\, \mathrm{det}\, A=1
\right\}
$$
for the Hamilton quaternion field $\mathbb{H}$ 
consisting of $2\times 2$ matrices. 
\end{proposition}

\textsc{Remark}: \textit{For arbitrary coefficients, $c_{L}, c_{R}\in\mathbb{C}$, 
the wave function} 
\begin{equation}
\psi=\psi_{0}+c_{L}\psi_{L}^{+}+c_{R}\psi_{R}^{+}
+U(c_{L}\psi_{L}^{+}+c_{R}\psi_{R}^{+})
\label{eq:wias-1}
\end{equation} 
\textit{is in the domain $D(H_{\alpha})$, 
where $\psi_{0}\in D(H_{0})$, and $\psi_{\sharp}^{\pm}$ were defined 
in Eqs.(\ref{eq:+eigenfunctions}) and (\ref{eq:-eigenfunctions}). 
Taking $0$ as the coefficient $c_{R}$ (resp. $c_{L}$) 
in the case where $U$ is non-diagonal, 
we have $\psi=\psi_{0}+c_{L}(\psi_{L}^{+}+\gamma_{1}\gamma_{3}\psi_{L}^{-}
-\gamma_{2}^{*}\gamma_{3}\psi_{R}^{-})$ 
(resp. $\psi=\psi_{0}+c_{R}(\psi_{R}^{+}+\gamma_{1}^{*}\gamma_{3}\psi_{R}^{-}
+\gamma_{2}\gamma_{3}\psi_{L}^{-})$). 
We set 
$k(z):=\sqrt{z-m}\sqrt{z+m}$ for $z\in\mathbb{C}$, 
where $\sqrt{z}$ is the branch of the complex square root with the cut 
along the non-negative real axis $\mathbb{R}_{+}$. 
The function $k(\cdot)$ is analytic in $\mathbb{C}\setminus 
[-m\, ,m]$, $\Im k(z)\ge 0$ for $z\in\mathbb{C}_{+}$, 
and $\Im k(z)\le 0$ for $z\in\mathbb{C}_{-}$ \cite{BMN02}. 
Then, since $\sqrt{1+m^{2}}=\mp ik(\pm i)$ and 
$|k(\pm i)|=\sqrt{1+m^{2}}$, we have 
$e^{\sqrt{1+m^{2}}\, x}=e^{\mp ik(\pm i)x}$ and $e^{-\sqrt{1+m^{2}}\, x}=e^{\pm ik(\pm i)x}$. 
Thus, the entry $\gamma_{1}$ is concerned with 
the reflection, and the entry $\gamma_{2}$ with the penetration.}

Now our second theorem is the following:  
\begin{theorem}
\label{theo:2}
\begin{enumerate}
\item[i)] Every diagonal $U\in U(2)$ 
has the following representation: 
There are complex numbers $\gamma_{L}, \gamma_{R}\in\mathbb{C}$ 
so that 
$$
U=\left(
\begin{array}{cc}
\gamma_{L} & 0 \\ 
0 & \gamma_{R}
\end{array}
\right)
\,\,\, \textrm{with}\,\,\, |\gamma_{L}|=|\gamma_{R}|=1.
$$
Then, for arbitrarily fixed 
$\gamma_{L}$ and $\gamma_{R}$ satisfying $|\gamma_{L}|=|\gamma_{R}|=1$, 
a necessary and sufficient condition for 
$D(H_U) = D(H_{\rho})$ is given by determining 
the vector $\rho\in\overline{\mathbb{R}}^{2}$ 
with the formulae:
\begin{enumerate}
\item[(L1)] For $\gamma_{L} \neq -1$, 
$\rho_{-}=\left(\tan\frac{\theta_{L}}{2}-m\right)/\sqrt{1+m^{2}}$, 
where $\theta_{L}:=\arg\gamma_{L}\in [0 , 2\pi)$.
\item[(L2)] For $\gamma_{L}=\, -1$, $\rho_{-}=+\infty$. 
\item[(R1)] For $\gamma_{R} \neq -1$, 
$\rho_{+}=\, -\left(\tan\frac{\theta_{R}}{2}-m\right)/\sqrt{1+m^{2}}$, 
where $\theta_{R}:=\arg \gamma_{R}\in [0 , 2\pi)$.
\item[(R2)] For $\gamma_{R}=\, -1$, $\rho_{+}=+\infty$.
\end{enumerate}
\item[ii)] Let $\mu$ be a constant defined by 
$\mu:=(1+im)/\sqrt{1+m^{2}}\in\mathbb{C}$. 
Every non-diagonal $U \in U(2)$ 
has the following representation: 
There are complex numbers $\gamma_{1}, \gamma_{2}, \gamma_{3}\in\mathbb{C}$ 
so that   
$$
U=\gamma_{3}
\left(
\begin{array}{cc}
\gamma_{1} & -\gamma_{2}^{*} \\ 
\gamma_{2} & \gamma_{1}^{*}
\end{array}
\right)
\,\,\, \textrm{with}\,\,\,  
|\gamma_{1}|^{2}+|\gamma_{2}|^{2}
=|\gamma_{3}|= 1\,\,\, \textrm{and}\,\,\, \gamma_{2} \neq 0. 
$$
Then, for arbitrarily fixed 
$\gamma_{1}$, $\gamma_{2}$, and $\gamma_{3}$ 
satisfying $|\gamma_{1}|^{2}+|\gamma_{2}|^{2}
=|\gamma_{3}|= 1$ and $\gamma_{2} \neq 0$, 
a necessary and sufficient condition 
for $D(H_U) = D(H_{\alpha})$ 
is given by determining the vector $\alpha\in\mathbb{C}^{4}$ with 
the formulae: 
\begin{equation}
\left\{ \begin{array}{l}
\alpha_{1} 
= i\gamma_{2}^{-1}\sqrt{1+m^{2}}
\left(
\Im(\gamma_{1}^{*}\mu)+\Im(\gamma_{3}^{*}\mu)
\right), \\ 
\alpha_{2} 
= \gamma_{2}^{-1}\sqrt{1+m^{2}}
\left(
\Re\gamma_{1}+\Re\gamma_{3}
\right), \\
\alpha_{3} 
= \gamma_{2}^{-1}\sqrt{1+m^{2}}
\left(
-\Re\gamma_{1}+\Re(\gamma_{3}^{*}\mu^{2})
\right), \\ 
\alpha_{4} 
= i\gamma_{2}^{-1}\sqrt{1+m^{2}}
\left(
\Im(\gamma_{1}\mu)+\Im(\gamma_{3}^{*}\mu)
\right).
\end{array}\right.
\label{eq:TJF1}
\end{equation}
\end{enumerate}
\end{theorem}
The proof of this theorem will appear in \S\ref{subsec:proof-theo-2}.

Since Proposition \ref{prop:von-Neumann} 
says that unitary operators $U\in U(2)$ determine 
all the self-adjoint extensions of the minimal 
Dirac operator, 
Theorem \ref{theo:2} gives the complete classification 
with the boundary conditions: 
\begin{corollary}
\label{cor:theo-2} 
The boundary conditions of all the self-adjoint extensions 
of the minimal Dirac operator $H_{0}$ 
can be classified under either one of 
the boundary conditions, (\ref{BC-rho}) and 
(\ref{BC-alpha}).
\end{corollary}

Theorem \ref{theo:2} gives the formulae showing how to construct 
the two parameters $\rho=(\rho_{+},\rho_{-})\in\overline{\mathbb{R}}^{2}$ 
(resp. the four parameters $\alpha=(\alpha_{1},\alpha_{2},\alpha_{3},\alpha_{4})
\in\mathbb{C}^{4}$) describing the boundary condition 
from the parameters, $(\gamma_{L},\gamma_{R})$ (resp. $(\gamma_{1},\gamma_{2},\gamma_{3})$), 
describing the unitary operator $U\in U(2)$ appearing in von Neumann's theory. 
We give the formulae conversely showing how to construct 
the parameters describing the unitary operator $U\in U(2)$ 
from the parameter family describing the boundary condition. 
 
Since Theorem \ref{theo:2} i) gives the one-to-one correspondence 
between the boundary condition (\ref{BC-rho}) and the parameters $(\gamma_{L},\gamma_{R})$ 
actually, we immediately have
\begin{enumerate}
\item[(L1')] $\gamma_{L}=\exp\left[ 2i\tan^{-1}\left( m+\sqrt{1+m^{2}}\rho_{-}\right)\right]$ 
if $\rho_{-}\in\mathbb{R}$, 
\item[(L2')] $\gamma_{L}=\, -1$ if $\rho_{-}=\infty$,  
\item[(R1')] $\gamma_{R}=\exp\left[ 2i\tan^{-1}\left( m-\sqrt{1+m^{2}}\rho_{+}\right)\right]$ 
if $\rho_{+}\in\mathbb{R}$, 
\item[(R2')] $\gamma_{R}=\, -1$ if $\rho_{+}=\infty$. 
\end{enumerate} 
The formulae for the other case are obtained 
using Propositions \ref{prop:03'} and \ref{prop:06}:
\begin{proposition}
\label{prop:inverse}
For every boundary matrix $B_{\alpha}$ 
with $\alpha\in\mathbb{C}^{4}$ in the class (\ref{Cl-alpha}), 
the corresponding non-diagonal $U\in U(2)=U(1)S\mathbb{H}$ 
is determined as: 
\begin{eqnarray}
\left\{ \begin{array}{l}
\gamma_{1}=\Gamma_{0}e^{-i(\theta-\pi/2)}
\left(
-\mu^{*}\alpha_{1}+\alpha_{2}-\alpha_{3}+\mu\alpha_{4}
\right), \\ 
\gamma_{2}=\frac{2}{\sqrt{1+m^{2}}}\Gamma_{0}e^{-i(\theta-\pi/2)}, \\ 
\gamma_{3}=\Gamma_{0}e^{-i(\theta-\pi/2)}\mu
\left(
\alpha_{1}+\mu^{*}\alpha_{2}+\mu\alpha_{3}+\alpha_{4}
\right)^{*},
\end{array}\right. 
\label{eq:TJF2}
\end{eqnarray}
where $\mu=(1+im)/\sqrt{1+m^{2}}$,  
$$
\Gamma_{0}=\left(
\frac{4}{1+m^{2}}
+|-\mu^{*}\alpha_{1}+\alpha_{2}-\alpha_{3}+\mu\alpha_{4}|^{2}
\right)^{-1/2},
$$
and $\theta$ is determined by following Proposition \ref{prop:06} as 
$\alpha_{j}=e^{i\theta}a_{j}$, $j=1, 4$, and 
$\alpha_{k}=ie^{i\theta}a_{k}$, $j=2, 3$.  
\end{proposition}

We will prove this proposition in \S\ref{subsection:prop-inverse}. 

The following proposition says that our $\alpha\in\mathbb{C}^{4}$ 
in the class (\ref{Cl-alpha}) is equivalent to 
Benvegn\`{u} and D\c{a}browski's four-parameter family, 
which shows how a phase factor appears 
in the boundary matrix: 
\begin{proposition}
\label{prop:06}
Let $\mathcal{A}$ be the set of all boundary matrices $B_{\alpha}$ 
for vectors $\alpha=(\alpha_{1},\alpha_{2},\alpha_{3},\alpha_{4})
\in\mathbb{C}^{4}$ in the class (\ref{Cl-alpha}). 
Then, $\alpha_{1}\ne 0$ or $\alpha_{3}\ne 0$. 
So, set $\theta\in \left[\left. 0 , 2\pi\right)\right.$, and 
$a_{1}, a_{2}, a_{3}, a_{4}\in\mathbb{R}$ as 
$$
\left\{ \begin{array}{l}
\theta:= \arg(\alpha_{1}/|\alpha_{1}|); \\  
a_{1}:= |\alpha_{1}|,\,\,\, 
a_{2}:=\, -i(\alpha_{1}\alpha_{2}^{*})^{*}/|\alpha_{1}|,\,\,\,    
a_{3}:=\, -i(\alpha_{1}\alpha_{3}^{*})^{*}/|\alpha_{1}|,\,\,\, 
a_{4}:= (\alpha_{1}\alpha_{4}^{*})^{*}/|\alpha_{1}|,
\end{array}\right.
$$
if $\alpha_{1}\ne 0$, and 
$$
\left\{ \begin{array}{l}
\theta:= \arg(-i\alpha_{3}/|\alpha_{3}|); \\  
a_{1}:= i\alpha_{1}\alpha_{3}^{*}/|\alpha_{3}|,\,\,\, 
a_{2}:= \alpha_{2}\alpha_{3}^{*}/|\alpha_{3}|,\,\,\,    
a_{3}:= |\alpha_{3}|,\,\,\, 
a_{4}:= i(\alpha_{3}\alpha_{4}^{*})^{*}/|\alpha_{3}|,
\end{array}\right.
$$
if $\alpha_{1}=0$.  
Then, $\mathcal{A}$ has the following representation: 
$$
\mathcal{A}=
\Biggl\{ e^{i\theta}
\left(
\begin{array}{cc}
a_{1} & ia_{2} \\ 
ia_{3} & a_{4}
\end{array}
\right)
\, \Biggl| \, 
\theta \in [0, 2\pi),\, 
a_{j} \in \mathbb{R},\, 
j=1, 2, 3, 4,\,\,\, 
\textrm{with}\,\,\, 
a_{1}a_{4}+a_{2}a_{3}=1
\Biggl\}.
$$
\end{proposition}

\textsc{Remark}: \textit{The Benvegn\`{u} and D\c{a}browski's four-parameter family, 
consisting of $A, B, C, D \in\mathbb{R}$ and $\omega\in\mathbb{C}$, 
as in Eq.(15) of Ref.\cite{BD94} is given by the correspondence, 
$\omega=e^{i\theta}$, $A=a_{1}, B=a_{2}, C=\, -a_{3}$, and $D=a_{4}$.}

\textit{
Meanwhile, in the case of the Schr\"{o}dinger particle living in 
our configuration space $\Omega_{\Lambda}$, 
the boundary matrix $B_{\alpha}$ making the boundary condition, 
$$
\begin{pmatrix}
\psi(+\Lambda) \\ 
\psi'(+\Lambda)
\end{pmatrix}
=B_{\alpha}
\begin{pmatrix}
\psi(-\Lambda) \\ 
\psi'(-\Lambda)
\end{pmatrix}, 
$$ 
has the four parameters satisfying $\alpha_{1}\alpha_{3}^{*},  
\alpha_{2}\alpha_{4}^{*}\in\mathbb{R}$ and $\alpha_{1}\alpha_{4}-\alpha_{2}\alpha_{3}=1$, 
and moreover,   
the set $\mathcal{A}$ has the following representation: 
$$
\mathcal{A}=
\Biggl\{ e^{i\theta}
\left(
\begin{array}{cc}
a_{1} & a_{2} \\ 
a_{3} & a_{4}
\end{array}
\right)
\, \Biggl| \, 
\theta \in [0, 2\pi),\, 
a_{j} \in \mathbb{R},\, 
j=1, 2, 3, 4,\,\,\, 
\textrm{with}\,\,\, 
a_{1}a_{4}-a_{2}a_{3}=1
\Biggl\}.
$$
For more details, see Proposition 2.6 of Ref.\cite{HK13-S}.  
}

Thus, Proposition \ref{prop:06}, together with Eqs.(\ref{eq:TJF1}), says 
that the Benvegn\`{u} and D\c{a}browski's four-parameter family 
can actually characterized by three parameters coming from von Neumann's theory: 
\begin{corollary}
\label{cor:BD->3}
The Benvegn\`{u} and D\c{a}browski's four-parameter family, 
consisting of $A, B, C, D \in\mathbb{R}$ and $\omega\in\mathbb{C}$, 
is characterized by three parameters, $\gamma_{1}, \gamma_{2}, \gamma_{3}\in\mathbb{C}$ 
with $|\gamma_{1}|^{2}+|\gamma_{2}|^{2}=|\gamma_{3}|=1$ 
and $\gamma_{2}\ne 0$. 
\end{corollary}

\section{Mathematical Idea of Tunnel-Junction Device for Spintronic Qubit}

In this section, we propose a mathematical idea for 
a tunnel-junction device for spintronic qubit. 
Of course, since we derive mathematically-theoretical possible mechanism 
from our simple toy model, 
we are not sure that the idea can be experimentally demonstrated. 
Even this toy model, however, tells us that 
we have to mind the effect of a phase coming from the boundary. 
We can see such an effect in the Andreev(-like) effects in 
more realistic cases Refs.\cite{A64,A65,TOD08}. 
Conversely, we may use the phase effect for a device. 
We are interested in the unit of a quantum device, 
consisting of a junction and two quantum wires 
such as in Fig.\ref{fig:unit} 
from the point of the view of quantum engineering. 
\begin{figure}[htbp]
\begin{center}
 \begin{minipage}{0.8\hsize}
  \begin{center}
   \includegraphics[width=60mm]{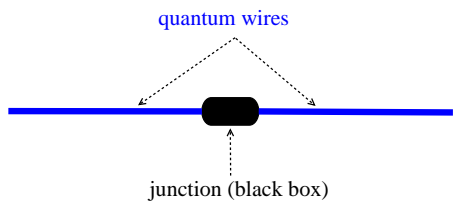}
  \end{center}
  \caption{\scriptsize 
The unit of our quantum device consists of 
the two quantum wires (blue solid lines)  
and the junction as the black box.}
  \label{fig:unit}
 \end{minipage}
\end{center}
\end{figure}
The combination of these units makes a quantum network. 
The junction is for controlling 
the information of qubit. 
The wires play a role of transporting the information.

We suppose that the energy of our unit has the Hamiltonian, 
$H_{\mathrm{wires}}+H_{\mathrm{junction}}+H_{\mathrm{interaction}}$, 
where $H_{\mathrm{wires}}$ is the Hamiltonian 
for the single electron living in the two wires, 
$H_{\mathrm{junction}}$ the Hamiltonian 
for the electron in the junction consisting of a physical object 
such as a quantum dot, and $H_{\mathrm{interaction}}$ describes 
the interaction between the wires and the junction. 
The Hamiltonians $H_{\mathrm{wires}}$ and $H_{\mathrm{junction}}$ 
should be observables in physics, and therefore, 
self-adjoint operators in mathematics then. 
We actually have to determine a concrete physical object for the junction 
to complete and realize our unit in the quantum engineering. 
But, in this paper, we regarded the junction as a black box 
so that the junction has 
mathematical, physical arbitrariness. 
Thus, we handled the Hamiltonian $H_{\mathrm{wires}}$ only, 
but we adopted proper boundary condition between the two wires 
and the junction instead of considering 
the Hamiltonian $H_{\mathrm{junction}}$ 
and the interaction $H_{\mathrm{interaction}}$ 
so that the Hamiltonian $H_{\mathrm{wires}}$ becomes 
observable, i.e., self-adjoint. 
The self-adjointness of the Hamiltonian $H_{\mathrm{wires}}$ 
is mathematically determined by a boundary condition 
of the wave functions on which the Hamiltonian $H_{\mathrm{wires}}$ acts. 
In addition to this, 
the boundary condition is uniquely determined by 
the quality and the shape of the boundary of a material 
of the wires in real physics.  
Thus, the wave functions have to satisfy the unit's own 
specific boundary condition 
to become the residents of the unit, 
otherwise they are ejected. 

Our one-to-one correspondence formulae show 
how the Benvegn\`{u} and D\c{a}browski's four-parameter family 
are concretely determined. 
Their four-parameter family shows 
how the phase factor appears and how the electron spin is affected at the boundary. 
Since von Neumann's theory gives the form of the wave functions, 
we can grasp how they pass through the junction. 
On the other hand, the boundary which is not characterized by 
the Benvegn\`{u} and D\c{a}browski's four-parameter family is 
the case where the wave functions never infiltrate the junction, 
which is characterized by the two parameters. 
As is well known, this case is also important, for instance, 
to demonstrate the Aharonov-Bohm effect experimentally \cite{Tonomura1,Tonomura2}. 
Thus, through our observation along with Benvegn\`{u} and D\c{a}browski's result, 
we understood that, for the wave functions which pass through the junction, 
the boundary condition has its own relation between 
the phase factor and the electron spin. 
The results may suggest a mathematical possibility of making 
a device for switching the channel of qubit. 
There is a case of the two units in Fig.\ref{fig:units}: 
\begin{figure}[htbp]
\begin{center}
 \begin{minipage}{0.8\hsize}
  \begin{center}
   \includegraphics[width=60mm]{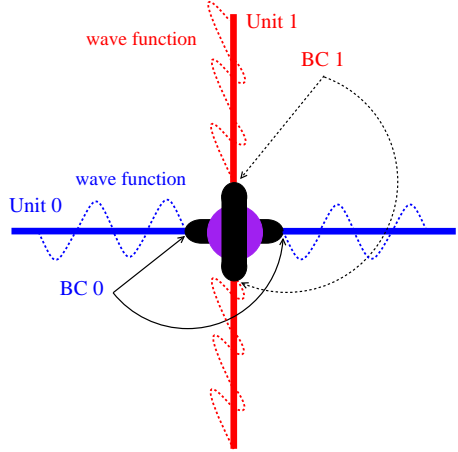}
  \end{center}
\vspace{-5mm}
  \caption{\scriptsize 
Qubit switch between the two units: 
The electron wave functions living in quantum wires (blue transversus 
solid lines) of Unit0 has the boundary condition BC0. 
If we can change the boundary condition from BC0 to another boundary 
condition BC1 in a gate (purple disc) connected the junctions, 
then the electron wave functions have to be residents 
in the quantum wires (red longitudinal solid lines) of Unit1.}
  \label{fig:units}
 \end{minipage}
\end{center}
\end{figure}
The Unit{\,}0 accepts the wave functions with 
the boundary condition BC{\,}0 only, 
and refuses the wave functions with another boundary condition. 
On the other hand, the Unit{\,}1 welcomes the wave functions with 
the boundary condition BC{\,}1 different from BC{\,}0, 
though it rejects the wave functions with the boundary condition BC{\,}0.  
For instance, as shown in Eqs.(\ref{eq:spin-flip}) 
and (\ref{eq:no-spin-flip}) below, 
we can use Unit{\,}0 for the channel without spin-flip 
and Unit{\,}1 for the channel with spin-flip 
as well as we can use the units for channels 
for the phase-shifted qubit. 
Thus, regarding Unit{\,}0 and Unit{\,}1 as a qubit, 
our switching device in Fig.\ref{fig:units} may play a role of 
\textit{quantum state transfer} from spintronic qubit. 
Here we should remember the experimental demonstration of 
the flying qubit \cite{UniTokyo12}, which is realized by 
the presence of an electron in either channel of the wire 
of an Aharonov-Bohm ring.

Theorem \ref{theo:2} shows the correspondence:  
\begin{eqnarray*} 
\textrm{$U\in U(2)$ is diagonal}\,\,\, 
&{\longleftrightarrow}&\,\,\, 
\textrm{(\ref{BC-rho})} \\ 
\textrm{$U\in U(2)$ is non-diagonal}\,\,\, 
&{\longleftrightarrow}&\,\,\,  
\textrm{(\ref{BC-alpha})}.
\qquad\qquad
\end{eqnarray*}
Theorem \ref{theo:2} assures us that \textit{there is no boundary condition 
which makes a self-adjoint extension but conditions, 
(\ref{BC-rho}) and (\ref{BC-alpha})}. 
These two conditions make the broad difference: 
The solitariness in the boundary condition (\ref{BC-rho}),
$$
(\textrm{left island})\qquad  
\left\{ \begin{array}{l}
i\rho_-\psi_{\uparrow}(-\Lambda)
=\psi_{\downarrow}(-\Lambda) 
\,\,\, \textrm{if\,\,\, $\rho_{-} \in \mathbb{R}$}, \\
\psi_{\uparrow}(-\Lambda)=0
\,\,\, \textrm{if\,\,\, $\rho_{-}=\infty$},
\end{array}\right.
$$ 
and 
$$ 
(\textrm{right island})\qquad  
\left\{ \begin{array}{l}
i\rho_{+}\psi_{\uparrow}(+\Lambda)
=\psi_{\downarrow}(+\Lambda) 
\,\,\, \textrm{if\,\,\, $\rho_{+} \in \mathbb{R}$}, \\
\psi_{\uparrow}(+\Lambda)=0 
\,\,\, \textrm{if\,\,\, $\rho_{+}=\infty$}.   
\end{array}\right. 
$$
Both of boundary conditions in the left island 
and the right one are independent of each other, 
which says that 
\textit{there is no interchange 
between the wave functions living in 
the left island $\Omega_{\Lambda,L}$ 
and those living in 
the right island $\Omega_{\Lambda,R}$} 
because the information of the wave functions 
never infiltrates the junction. 
In addition, \textit{no special phase factor but $\pm\pi/2$ 
appears in this boundary condition then}. 

On the other hand, according to Eq.(15) of Ref.\cite{BD94} described by 
the Benvegn\`{u} and D\c{a}browski's four-parameter family 
with the representation Proposition \ref{prop:06}, 
the boundary condition (\ref{BC-alpha}) 
shows \textit{how the wave functions living in the left island and 
the those living in the right island make 
interchange between each other}, 
and \textit{how the electron spin is affected by the phase 
factor at the boundaries}: 
\begin{equation}
\left(
\begin{array}{c}
\psi_{\uparrow}(+\Lambda) \\ 
\psi_{\downarrow}(+\Lambda) \\ 
\end{array}
\right) 
=
\left(
\begin{array}{c}
e^{i\theta}a_{1}\psi_{\uparrow}(-\Lambda)
+e^{i(\theta+\pi/2)}a_{2}\psi_{\downarrow}(-\Lambda) \\ 
e^{i(\theta+\pi/2)}a_{3}\psi_{\uparrow}(-\Lambda)
+e^{i\theta}a_{4}\psi_{\downarrow}(-\Lambda) 
\end{array}
\right)
\label{eq:TJF}
\end{equation}
for some $a_{j}\in\mathbb{R}$, 
$j= 1, \cdots, 4$, with 
$a_{1}a_{4}+a_{2}a_{3}=1$. 
The wave function (\ref{eq:wias-1}) determines 
four parameters $\alpha=(\alpha_{1},\alpha_{2},\alpha_{3},\alpha_{4})\in\mathbb{C}^{4}$ 
through Eq.(\ref{eq:TJF1}). 
If $a_{j}\in\mathbb{R}$, $j=1, \cdots, 4$, 
satisfy $a_{1}=a_{4}=0$ and $a_{2}a_{3}=1$, we obtain the boundary condition 
so that the \textit{spin-flip} with the \textit{phase factor} $e^{i(\theta+\pi/2)}$ 
for an arbitrary $\theta\in [0,2\pi)$ takes place, 
that is, the up-spin and the down-spin interchange 
with each other: 
\begin{equation}
\left(
\begin{array}{c}
\psi_{\uparrow}(+\Lambda) \\ 
\psi_{\downarrow}(+\Lambda) \\ 
\end{array}
\right) 
=e^{i(\theta+\pi/2)}
\left(
\begin{array}{c}
a_{2}\psi_{\downarrow}(-\Lambda) \\ 
a_{3}\psi_{\uparrow}(-\Lambda)
\end{array}
\right).
\label{eq:spin-flip}
\end{equation}
Meanwhile, if $a_{j}\in\mathbb{R}$, $j=1, \cdots, 4$, 
satisfy $a_{1}a_{4}=1$ and $a_{2}=a_{3}=0$, we obtain the boundary condition 
so that the phase factor $e^{i\theta}$ appears 
for an arbitrary $\theta\in [0,2\pi)$ 
but the spin-flip does not take place: 
\begin{equation}
\left(
\begin{array}{c}
\psi_{\uparrow}(+\Lambda) \\ 
\psi_{\downarrow}(+\Lambda) \\ 
\end{array}
\right) 
=e^{i\theta}
\left(
\begin{array}{c}
a_{1}\psi_{\uparrow}(-\Lambda) \\ 
a_{4}\psi_{\downarrow}(-\Lambda)
\end{array}
\right).
\label{eq:no-spin-flip}
\end{equation}

In Fig.\ref{fig:units}, for instance, 
let us employ Eq.(\ref{eq:no-spin-flip}) 
with $a_{1}=a_{4}=1$ and $\theta=0$ for Unit{\,}0, 
and Eq.(\ref{eq:spin-flip}) with $a_{2}=a_{3}=1$ 
and $\theta=\, -\pi/2$ for Unit{\,}1, respectively. 
We set the Pauli-X gate in the disc of junctions. 
Then, the residence of the wave functions living in Unit{\,}0 
is switched to Unit{\,}1 after the Pauli-X gate operation. 
That is, we have a spin-based switching device for qubit. 
Thus, there is a possibility that we can use this switching device 
for quantum state transfer from spintronic qubit 
regarding Unit{\,}0 and Unit{\,}1 as qubit. 
In the same way, if we employ Eq.(\ref{eq:no-spin-flip}) with 
$a_{1}=a_{4}=1$ and $0<\theta<2\pi$ for Unit{\,}1 instead, 
we can make a phase-based switching device for qubit. 
This means that we may control the qubit consisting of Unit{\,}0 and 
Unit{\,}1 through the phase factor $\theta$. 
We note that both the spin-flip gate operation 
and the phase-shift gate operation 
had been demonstrated in the experiment 
for the spin state of an electron-hole 
pair in a semiconductor quantum dot \cite{GBPSSPFRS10}.

\section{Proofs of Main Results} 

We will give individual proofs of our main results.

\subsection{Proof of Proposition \ref{prop:deficiency-indices}}
\label{subsec:proof-deficiency-indices}

We now prove Proposition \ref{prop:deficiency-indices}. 
Let $\psi = {^t}(\psi_{\uparrow},\psi_{\downarrow})$ be in 
the deficiency subspace $\mathcal{K}_{\pm}(H_0)$, i.e., 
$\psi\in\mathcal{K}_{+}(H_0)$ or $\psi\in\mathcal{K}_{-}(H_0)$.  
Then, since $H_0^{*}\psi=\pm i\psi$, 
Proposition \ref{prop:adjoint-operator} gives us 
the following differential equation: 
\begin{equation}
\left(
\begin{array}{c}
\psi_{\uparrow}{\,}' \\
\psi_{\downarrow}{\,}' 
\end{array}
\right)
=
\left(
\begin{array}{cc}
0 & (\mp 1 + im) \\ 
(\mp 1 - im) & 0 
\end{array}
\right)
\left(
\begin{array}{c}
\psi_{\uparrow} \\
\psi_{\downarrow} 
\end{array}
\right). 
\label{eq:proof-5-1-1}
\end{equation}
We note that $\mathcal{K}_{\pm}(H_{0})\subset 
\mathcal{AC}(\overline{\Omega_{\Lambda}})$. 
So, according to the general theory 
of differential equation, 
every solution $\psi$ 
of Eq.(\ref{eq:proof-5-1-1}) in 
$\mathcal{AC}(\overline{\Omega_{\Lambda}})$ 
is respectively written as
\begin{equation}
\left\{ \begin{array}{l}
\psi=c_{L}^{+}\psi_{L}^{+}+c_{R}^{+}\psi_{R}^{+},\,\,\, 
c_{L}^{+}, c_{R}^{+}\in\mathbb{C},\,\,\,  
 \textrm{if $\psi\in\mathcal{K}_{+}(H_{0})$,} \\ 
\psi=c_{L}^{-}\psi_{L}^{-}+c_{R}^{-}\psi_{R}^{-},\,\,\, 
c_{L}^{-}, c_{R}^{-}\in\mathbb{C},\,\,\,  
 \textrm{if $\psi\in\mathcal{K}_{-}(H_{0})$.}
\end{array}\right.
\label{eq:solution}
\end{equation} 
It follows from this representation that $n_{\pm}(H_{0})=2$ 
because the functions $\psi_{L}^{\sharp}$ and $\psi_{R}^{\sharp}$ 
mutually intersect orthogonally 
in the Hilbert space $L^{2}(\Omega_{\Lambda})$. 
The existence of self-adjoint extensions 
follows from Proposition \ref{prop:von-Neumann}.  

\subsection{Proof of Theorem \ref{theo:1}} 
\label{proof-theo-1}

To prove Theorem \ref{theo:1} 
we prepare the following lemma here. 

\begin{lemma}
\label{lem:02}
\begin{enumerate}
\item[i)] Let $a_{1}, a_{2}, b_{1}, b_{2}$ be arbitrary complex numbers.  
\begin{enumerate}
\item[(i-1)] For every $\rho=(\rho_{+},\rho_{-})$ 
with $|\rho_{\pm}|<\infty$, 
there is a wave function $\psi \in D(H_\rho)$ 
so that $\psi_{\uparrow}(+\Lambda)=a_{1}$ and 
$\psi_{\uparrow}(-\Lambda)=a_{2}$. 
\item[(i-2)] For every $\rho=(\rho_{+},\rho_{-})$ 
with $|\rho_{+}|<\infty$ and $\rho_{-}=+\infty$, 
there is a wave function $\psi \in D(H_\rho)$ 
so that $\psi_{\uparrow}(+\Lambda)=a_{1}$ and 
$\psi_{\downarrow}(-\Lambda)=b_{1}$. 
\item[(i-3)] For every $\rho=(\rho_{+},\rho_{-})$ 
with $\rho_{+}=+\infty$ and $|\rho_{-}|<\infty$, 
there is a wave function $\psi \in D(H_\rho)$ 
so that $\psi_{\uparrow}(-\Lambda)=a_{2}$ and 
$\psi_{\downarrow}(+\Lambda)=b_{2}$. 
\item[(i-4)] For every $\rho=(\rho_{+},\rho_{-})$ 
with $\rho_{\pm}=+\infty$, 
there is a wave function $\psi \in D(H_\rho)$ 
so that $\psi_{\downarrow}(-\Lambda)=b_{1}$ and 
$\psi_{\downarrow}(+\Lambda)=b_{2}$. 
\end{enumerate}  
\item[ii)] For arbitrary vector ${{}^{t}}(a_{1}, a_{2})
\in\mathbb{C}^{2}$, there is a wave function $\psi \in D(H_\alpha)$ 
so that ${{}^{t}}(\psi_{\uparrow}(-\Lambda) , \psi_{\downarrow}(-\Lambda))
={{}^{t}}(a_{1}, a_{2})$. 
\end{enumerate}
\end{lemma}

\begin{proof}
We embed the function spaces $\mathcal{AC}(\overline{\Omega_{\Lambda,L}})$ 
and $\mathcal{AC}(\overline{\Omega_{\Lambda,R}})$ in 
the function space $\mathcal{AC}(\overline{\Omega_{\Lambda}})$ 
in the following: 
for every $\psi\in\mathcal{AC}(\overline{\Omega_{\Lambda,L}})$, 
we expand the function $\psi$ as $\psi(x)=0$ 
for $x \in \overline{\Omega_{\Lambda,R}}$ and regard 
the function $\psi$ as the function on $\overline{\Omega_{\Lambda}}$. 
We employ the same expansion for 
functions in $\mathcal{AC}(\overline{\Omega_{\Lambda,R}})$. 

i) It is not so difficult to show this part. 
Let $\rho$ be in $\mathbb{R}^{2}$. 
Fix an arbitrary function 
$f\in AC^{1}(\overline{\Omega_{\Lambda,R}})$ 
with $f(+\Lambda)\ne 0$, and take it. 
For an arbitrary number $a_{1}\in\mathbb{C}$ 
we define functions $\psi_{R}$ by 
$\psi_{R}:=(a_{1}/f(+\Lambda))\, 
{}^{t}\!\left( f , i\rho_{+}f\right) 
\in\mathcal{AC}(\overline{\Omega_{\Lambda,R}})$.  
Similarly, take a function 
$g\in AC^{1}(\overline{\Omega_{\Lambda,L}})$ 
with $g(-\Lambda)\ne 0$. 
For an arbitrary number $a_{2}\in\mathbb{C}$ 
we define functions $\psi_{L}$ by 
$\psi_{L}:=(a_{2}/g(-\Lambda))\, 
{}^{t}\!\left( g , i\rho_{-}g\right)
\in\mathcal{AC}(\overline{\Omega_{\Lambda,L}})$.  

In the case where $|\rho_{\pm}|<\infty$, define a function $\psi\in 
\mathcal{AC}(\overline{\Omega_{\Lambda}})$ 
by $\psi:=\psi_{L}+\psi_{R}$. 
In the case where $|\rho_{+}|<\infty$ and $\rho_{-}=+\infty$, 
define a function $\psi\in 
\mathcal{AC}(\overline{\Omega_{\Lambda}})$ by 
$\psi:=\psi_{R}+(b_{1}/g(-\Lambda))\, 
{}^{t}\left( 0 , g\right)$. 
In the case where $\rho_{+}=+\infty$ and $|\rho_{-}|<\infty$, 
define a function $\psi\in \mathcal{AC}(\overline{\Omega_{\Lambda}})$ by 
$\psi:=\psi_{L}+(b_{2}/f(+\Lambda))\, 
{}^{t}\!\left( 0 , f\right)$.
In the case where $\rho_{\pm}=+\infty$, 
define a function $\psi\in 
\mathcal{AC}(\overline{\Omega_{\Lambda}})$ by 
$\psi:=(b_{1}/g(-\Lambda))\, 
{}^{t}\!\left( 0 , g\right)
+(b_{2}/f(+\Lambda))\, 
{}^{t}\!\left( 0 , f\right)$.  
Then, we obtain our desired wave function $\psi$. 

ii) Take functions $f, g\in AC^{1}(\overline{\Omega_{\Lambda,L}})$ 
and $h, k\in AC^{1}(\overline{\Omega_{\Lambda,R}})$ 
with $f(-\Lambda)\ne 0$, $g(-\Lambda)=0$, 
$h(+\Lambda)\ne 0$, and $k(+\Lambda)\ne 0$. 
For arbitrary vector ${{}^{t}}(a_{1}, a_{2})
\in\mathbb{C}^{2}$ we define functions $\varphi_{L}$ and 
$\varphi_{R}$ by 
$\varphi_{L}:=(a_{1}/f(-\Lambda))\, 
{}^{t}\!\left( f , g\right)
\in\mathcal{AC}(\overline{\Omega_{\Lambda,L}})$  
and 
$\varphi_{R}:={}^{t}\left(
(a_{1}\alpha_{1}/h(+\Lambda))h , 
(a_{1}\alpha_{3}/k(+\Lambda))k\right)
\in\mathcal{AC}(\overline{\Omega_{\Lambda,R}})$, 
respectively. 
Define the function $\varphi$ by 
$\varphi:=\varphi_{L}+\varphi_{R}$. 
Then, we reach the function $\varphi$ satisfying 
$\varphi\in D(H_{\alpha})$ with $\varphi(-\Lambda)={{}^{t}}( a_{1} , 0)$. 
In the same way, we can obtain a function 
$\phi$ satisfying $\phi\in D(H_{\alpha})$ with 
$\phi(-\Lambda)={{}^{t}}(0 , a_{2})$. 
Therefore, defining the function $\psi$ by 
$\psi:=\varphi + \phi$, this function is 
our desired one. \qed    
\end{proof}

We here prove Theorem \ref{theo:1}. 
Using integration by parts, 
for every $\psi,\phi \in D(H_{0}^{*})$ 
we have the following equation: 
\begin{eqnarray}
&{}& \langle H_{0}^{*}\psi\, |\,
\phi\rangle_{L^{2}(\Omega_{\Lambda})\oplus L^{2}(\Omega_{\Lambda})}
-\langle\psi\, |\, H_{0}^{*}\phi
\rangle_{L^{2}(\Omega_{\Lambda})\oplus L^{2}(\Omega_{\Lambda})} 
\label{eq:theo-1-1} \\ 
&=&\, -i \bigg\{\psi_{\uparrow}(+\Lambda)^{*}\phi_{\downarrow}(+\Lambda) 
+\psi_{\downarrow}(+\Lambda)^{*}\phi_{\uparrow}(+\Lambda) 
\nonumber \\
&{}&\qquad\quad 
-\psi_{\uparrow}(-\Lambda)^{*}\phi_{\downarrow}(-\Lambda) 
-\psi_{\downarrow}(-\Lambda)^{*}\phi_{\uparrow}(-\Lambda)\bigg\}.
\nonumber 
\end{eqnarray}

i) We prove our statement in the case where 
$\rho\in\mathbb{R}^{2}$ only. 
It is clear that $H_{0}\subset H_{\rho}$. 

First up, we show $H_{\rho}\subset H_{\rho}^{*}$. 
Since $H_{\rho}\subset H_{0}^{*}$, 
Eq.(\ref{eq:theo-1-1}) leads to the equation, 
\begin{eqnarray*}
&{}&\langle H_{\rho}\psi\,|\,
\phi\rangle_{L^{2}(\Omega_{\Lambda})\oplus L^{2}(\Omega_{\Lambda})}
-\langle\psi\, |\, H_{\rho}
\phi\rangle_{L^{2}(\Omega_{\Lambda})\oplus L^{2}(\Omega_{\Lambda})} 
\nonumber \\
&=&\, -i\bigg\{\psi_{\uparrow}(+\Lambda)^{*}
(i\rho_{+}\phi_{\uparrow}(+\Lambda))
+(i\rho_{+}\psi_{\uparrow}(+\Lambda))^{*}\phi_{\uparrow}(+\Lambda) \\
&{}&\qquad\quad 
- \psi_{\uparrow}(-\Lambda)^{*}(i\rho_{-}\phi_{\uparrow}(-\Lambda)) 
- (i\rho_{-}\psi_{\uparrow}(-\Lambda))^{*}\phi_{\uparrow}(-\Lambda) \bigg\} 
= 0
\end{eqnarray*}
for every $\psi, \phi \in D(H_{\rho})$. 
This means that $H_{\rho}$ is symmetric, i.e., 
$H_{\rho}\subset H_{\rho}^{*}$. 

Next, we show $H_{\rho}\supset H_{\rho}^{*}$. 
Based on Lemma \ref{lem:02} (i-1), 
for arbitrary vector ${{}^{t}}(a_{1}, a_{2})
\in\mathbb{C}^{2}$, we employ the wave function $\psi \in D(H_{\rho})$ 
so that $\psi_{\uparrow}(+\Lambda)=a_{1}^{*}$ and 
$\psi_{\uparrow}(-\Lambda)=a_{2}^{*}$. 
Using the definition of the adjoint operator, 
the fact that $H_{\rho}\subset H_{\rho}^{*}\subset H_{0}^{*}$, and 
Eq.(\ref{eq:theo-1-1}), we have the equation, 
\begin{eqnarray*}
0 &{=}& 
\langle H_{\rho}\psi\, |\,
\phi\rangle_{L^{2}(\Omega_{\Lambda})\oplus L^{2}(\Omega_{\Lambda})}
-\langle\psi\, |\, H_{\rho}^{*}
\phi\rangle_{L^{2}(\Omega_{\Lambda})\oplus L^{2}(\Omega_{\Lambda})} \\ 
&{=}&\, -i\bigg\{ a_{1}
(\phi_{\downarrow}(+\Lambda) -i\rho_{+}\phi_{\uparrow}(+\Lambda)) 
- a_{2}(\phi_{\downarrow}(-\Lambda)- i\rho_{-}\phi_{\uparrow}(-\Lambda)) \bigg\}  
\end{eqnarray*}
for every $\phi \in D(H_{\rho}^{*})$. 
Because the complex numbers $a_{1}$ and $a_{2}$ were arbitrarily, 
we have the equations: 
$\phi_{\downarrow}(+\Lambda) -i\rho_{+}\phi_{\uparrow}(+\Lambda)=0$ and 
$\phi_{\downarrow}(-\Lambda)- i\rho_{-}\phi_{\uparrow}(-\Lambda)$ 
for every $\phi \in D(H_{\rho}^{*})$. 
These conditions say that $\phi \in D(H_{\rho})$, namely, 
$D(H_{\rho}^{*})\subset D(H_{\rho})$ and thus $H_{\rho}^{*}\subset H_{\rho}$. 
Therefore, we have showed that the Dirac operator $H_{\rho}$ 
is a self-adjoint extension 
of the minimal Dirac operator: 
$H_{0}\subset H_{\rho}=H_{\rho}^{*}$.

In the same way, we can prove our statement for the case where 
$|\rho_{+}|=\infty$ or $|\rho_{-}|=\infty$ 
with the help of Lemma \ref{lem:02} (i-2)--(i-4). 

ii) Let us fix an arbitrary vector $\alpha\in\mathbb{C}^{4}$ 
in the class (\ref{Cl-alpha}). 
It is clear that $H_{0}\subset H_{\alpha}$. 
By Eq.(\ref{eq:theo-1-1}) together with the fact that 
$H_{\alpha}\subset H_{0}^{*}$, we have the equation, 
\begin{eqnarray*}
&{}&\langle H_{\alpha}\psi\, |\, 
\phi\rangle_{L^{2}(\Omega_{\Lambda})\oplus L^{2}(\Omega_{\Lambda})}
-\langle\psi\, |\, 
H_{\alpha}\phi\rangle_{L^{2}(\Omega_{\Lambda})\oplus L^{2}(\Omega_{\Lambda})} \\
&{=}&\, -i\biggl\{
(\alpha_{1}^{*}\alpha_{3}
+\alpha_{1}\alpha_{3}^{*})
\psi_{\uparrow}(-\Lambda)^{*}\phi_{\uparrow}(-\Lambda)
+(\alpha_{2}^{*}\alpha_{4}+\alpha_{2}\alpha_{4}^{*})
\psi_{\downarrow}(-\Lambda)^{*}\phi_{\downarrow}(-\Lambda) \\ 
&{}&\qquad\quad 
+(\alpha_{1}^{*}\alpha_{4}+\alpha_{2}\alpha_{3}^{*}-1)
\psi_{\uparrow}(-\Lambda)^{*}\phi_{\downarrow}(-\Lambda) \\
&{}&\qquad\qquad 
+(\alpha_{1}\alpha_{4}^{*}
+\alpha_{2}^{*}\alpha_{3}-1)
\psi_\downarrow(-\Lambda)^{*}\phi_{\uparrow}(-\Lambda)
\biggr\} = 0
\end{eqnarray*}
for every $\psi, \phi \in D(H_{\alpha})$. 
This means that $H_{\alpha}$ is symmetric, i.e., 
$H_{\alpha}\subset H_{\alpha}^{*}$.

Based on Lemma \ref{lem:02} ii), 
for arbitrary vector ${{}^{t}}(a_{1}, a_{2})
\in\mathbb{C}^{2}$, we employ the wave function $\psi \in D(H_{\alpha})$ 
so that ${{}^{t}}(\psi_{\uparrow}(-\Lambda) , \psi_{\downarrow}(-\Lambda))
={{}^{t}}(a_{1}^{*}, a_{2}^{*})$. 
Using the definition of the adjoint operator, 
the fact that $H_{\alpha}\subset H_{\alpha}^{*}\subset H_{0}^{*}$, and 
Eq.(\ref{eq:theo-1-1}), we have the equation, 
\begin{eqnarray*}
0 &{=}& 
\langle H_{\alpha}\psi\, |\, 
\phi\rangle_{L^{2}(\Omega_{\Lambda})\oplus L^{2}(\Omega_{\Lambda})} 
-\langle\psi\, |\, 
H_{\alpha}^{*}\phi\rangle_{L^{2}(\Omega_{\Lambda})\oplus L^{2}(\Omega_{\Lambda})}  \\ 
&{=}&\, -i\biggl\{ 
(\alpha_{1}^{*}\phi_{\downarrow}(+\Lambda)
+\alpha_{3}^{*}\phi_{\uparrow}(+\Lambda)
-\phi_{\downarrow}(-\Lambda))a_{1} \\
&{}&\qquad\quad 
+(\alpha_{2}^{*}\phi_{\downarrow}(+\Lambda)
+\alpha_{4}^{*}\phi_{\uparrow}(+\Lambda)
-\phi_{\uparrow}(-\Lambda))a_{2}
\biggr\} 
\end{eqnarray*}
for every $\phi \in D(H_{\alpha}^{*})$. 
Because the complex numbers $a_{1}$ and $a_{2}$ were arbitrarily, 
we have the equations: 
$\alpha_{1}^{*}\phi_{\downarrow}(+\Lambda)
+\alpha_{3}^{*}\phi_{\uparrow}(+\Lambda)
-\phi_{\downarrow}(-\Lambda)=0$ 
and 
$\alpha_{2}^{*}\phi_{\downarrow}(+\Lambda)
+\alpha_{4}^{*}\phi_{\uparrow}(+\Lambda)
-\phi_{\uparrow}(-\Lambda)=0$ 
for every $\phi \in D(H_{\alpha}^{*})$. 
We here note that the immediate computation 
leads to the entries of the inverse boundary matrix as 
$$
B_{\alpha}^{-1}
=
\left(
\begin{array}{cc}
\alpha_{4}^{*} & \alpha_{2}^{*} \\ 
\alpha_{3}^{*} & \alpha_{1}^{*}
\end{array}
\right)
$$
since $\alpha=(\alpha_{1},\alpha_{2},\alpha_{3},\alpha_{4})$ 
satisfies the conditions of the class (\ref{Cl-alpha}). 
Hence it follows that $\phi \in D(H_{\alpha})$, namely, 
$D(H_{\alpha}^{*})\subset D(H_{\alpha})$ and thus $H_{\alpha}^{*}\subset H_{\alpha}$. 
Therefore, we have showed that the Dirac operator $H_{\alpha}$ 
is a self-adjoint extension 
of the minimal Dirac operator: 
$H_{0}\subset H_{\alpha}=H_{\alpha}^{*}$.

\subsection{Proof of Proposition \ref{prop:03'}} 
\label{subsec:proof-prop-03'} 

First up, we rewrite $SU(2)$ in terms of the electron spin, 
that is, in terms of the Hamilton quaternion field 
spanned by the Pauli spin matrices: 
\begin{lemma}
\label{lem:03} 
The special unitary group $SU(2)$ has the following representation: 
$$
SU(2)=S\mathbb{H}=
\left\{
\left(
\begin{array}{cc}
\alpha & - \beta^{*} \\ 
\beta & \alpha^{*}
\end{array}
\right) 
\in \mathbb{H}\, \bigg|\, 
\alpha, \beta \in \mathbb{C},\,\,\, 
|\alpha|^{2}+|\beta|^{2} = 1
\right\}. 
$$
\end{lemma}

\begin{proof} 
In this proof we use the following representation. 
Introducing the argument $\theta_{j}\in [0 , 2\pi)$ 
of each entry $u_{j}$ of the matrix $U\in M_{2}(\mathbb{C})$, 
we represent $U$ as 
\begin{equation}
U=
\left(
\begin{array}{cc}
u_{1} & u_{2} \\ 
u_{3} & u_{4} 
\end{array}
\right)
\,\,\, 
\textrm{with}\,\,\, 
u_{j} =|u_{j}|e^{i\theta_{j}},\,\,\, 
j=1,\cdots,4.  
\label{eq:representation-matrix}
\end{equation}

Let us handle an arbitrary $U\in U(2)$ for a while. 
The unitarity of $U$ leads to 
the equations: 
\begin{eqnarray}
I_{\mathbb{C}^{2}}&=UU^{*}
=\left(
\begin{array}{cc}
|u_{1}|^{2}+|u_{2}|^{2} & 
u_{1}u_{3}^{*}+u_{2}u_{4}^{*} \\ 
u_{1}^{*}u_{3}+u_{2}^{*}u_{4} 
& |u_{3}|^{2}+|u_{4}|^{2}
\end{array}
\right). 
\label{eq:unitary-eq01} \\ 
I_{\mathbb{C}^{2}}&= U^*U
=
\left(
\begin{array}{cc}
|u_{1}|^{2}+|u_{3}|^{2} 
& u_{1}^{*}u_{2}+u_{3}^{*}u_{4} \\
u_{1}u_{2}^{*}+u_{3}u_{4}^{*} 
& |u_{2}|^{2}+|u_{4}|^{2}
\end{array}
\right). 
\label{eq:unitary-eq02}
\end{eqnarray}
Comparing the diagonal entries in the first row 
and the first column of both sides of Eq.(\ref{eq:unitary-eq02}), 
we have the equality:
\begin{equation}
|u_{1}|^{2}+|u_{3}|^{2}=1. 
\label{eq:5-3-1}
\end{equation}
In addition, we similarly have the equality, $|u_{1}|^{2}+|u_{2}|^{2}=1$, 
by Eq.(\ref{eq:unitary-eq01}). 
Thus, it follows from these two equalities that
\begin{equation}
|u_{2}|=|u_{3}|.  
\label{eq:5-3-2}
\end{equation}
In the same way, using the equalities, $|u_{1}|^{2}+|u_{2}|^{2}=1$ 
and $|u_{2}|^{2}+|u_{4}|^{2}=1$, we reach the equality: 
\begin{equation}
|u_{1}|=|u_{4}|. 
\label{eq:5-3-3}
\end{equation}
Comparing the off-diagonal entries in the first row 
and the second column of both sides of in Eq.(\ref{eq:unitary-eq01}), 
and using Eqs.(\ref{eq:5-3-2}) and (\ref{eq:5-3-3}),  
we have the equation, 
$$|u_{1}||u_{3}|\left( 
e^{i(\theta_{1}-\theta_{3})}+
e^{i(\theta_{2}-\theta_{4})}
\right)=0.
$$ 
Multiplying the both sides of this by $e^{i(\theta_{3}+\theta_{4})}$, 
we have 
\begin{equation}
|u_{1}||u_{3}|
\left( 
e^{i(\theta_{1}+\theta_{4})}+
e^{i(\theta_{2}+\theta_{3})}
\right)=0. 
\label{eq:5-3-4}
\end{equation}
Thus, we have derived the conditions (\ref{eq:5-3-1}), (\ref{eq:5-3-2}), 
(\ref{eq:5-3-3}), and (\ref{eq:5-3-4}) 
from the condition $U\in U(2)$. 

Conversely, it is easy to check that 
Eqs.(\ref{eq:5-3-1}), (\ref{eq:5-3-2}), 
(\ref{eq:5-3-3}), and (\ref{eq:5-3-4}) 
bring us to the condition  $U\in U(2)$. 
Consequently, the condition  $U\in U(2)$ is equivalent to 
the conditions (\ref{eq:5-3-1}), (\ref{eq:5-3-2}), 
(\ref{eq:5-3-3}), and (\ref{eq:5-3-4}): 
\begin{equation}
U\in U(2) 
\Longleftrightarrow 
\textrm{(\ref{eq:5-3-1}), (\ref{eq:5-3-2}), 
(\ref{eq:5-3-3}), and (\ref{eq:5-3-4})}. 
\label{eq:5-3-8'}
\end{equation}

Let us consider the case where $U\in SU(2)$ from now on. 
Then, we have an extra condition: 
\begin{equation}
\mathrm{det}\, U=
u_{1}u_{4}-u_{2}u_{3}=1. 
\label{eq:unitary-eq03}
\end{equation}
Combining Eq.(\ref{eq:unitary-eq03}) 
with Eqs.(\ref{eq:5-3-1})--(\ref{eq:5-3-3}) leads to 
the equations, 
$$
|u_{1}|^{2}+|u_{3}|^{2}=1=
u_{1}u_{4}-u_{2}u_{3}=
|u_{1}|^{2}e^{i(\theta_{1}+\theta_{4})}
-|u_{3}|^{2}e^{i(\theta_{2}+\theta_{3})},
$$
which implies that 
\begin{equation}
|u_{3}|^{2}\left(1+e^{i(\theta_{2}+\theta_{3})}\right)
=
|u_{1}|^{2}\left(e^{i(\theta_{1}+\theta_{4})}-1\right). 
\label{eq:5-3-5}
\end{equation}

Assume that $u_{1}u_{3}\ne 0$ now. 
Then, Eq.(\ref{eq:5-3-4}) brings us to the equation: 
\begin{equation}
e^{i(\theta_{1}+\theta_{4})}+e^{i(\theta_{2}+\theta_{3})}=0. 
\label{eq:5-3-6}
\end{equation}
Eqs.(\ref{eq:5-3-5}) and (\ref{eq:5-3-6}) say that 
$|u_{3}|^{2}\left(1-e^{i(\theta_{1}+\theta_{4})}\right)
=
|u_{1}|^{2}\left(e^{i(\theta_{1}+\theta_{4})}-1\right)$. 
Suppose that $1-e^{i(\theta_{1}+\theta_{4})}\ne 0$ here. 
Then, the above equation leads to the relation,  
$|u_{3}|^{2}=\, -|u_{1}|^{2}<0$. 
This is a contradiction. 
Thus, by the reductio ad absurdum, 
we know that $e^{i(\theta_{1}+\theta_{4})}=1$ and therefore 
$e^{i(\theta_{2}+\theta_{3})}=\, -1$ by 
Eq.(\ref{eq:5-3-6}). 

On the other hand, assume that $u_{1}u_{3}=0$. 
Then, we have the condition, $u_{1}\ne 0$ or $u_{3}\ne 0$, 
by Eq.(\ref{eq:5-3-1}). 
In the case where $u_{1}\ne 0$, 
since we have the equality $u_{3}=0$, 
the equation $e^{i(\theta_{1}+\theta_{4})}=1$ 
comes up from Eq.(\ref{eq:5-3-5}). 
In the case $u_{3}\ne 0$, similarly,  
the equation $e^{i(\theta_{2}+\theta_{3})}=\, -1$ 
is derived from Eq.(\ref{eq:5-3-5}). 
These arguments make us realize that
\begin{equation}
\left\{ \begin{array}{l}
\textrm{if $u_{1}\ne 0$, 
then $e^{i(\theta_{1}+\theta_{4})}=1$}; \\ 
\textrm{if $u_{3}\ne 0$, 
then $e^{i(\theta_{2}+\theta_{3})}=\, -1$}.
\end{array}\right.
\label{eq:5-3-7}
\end{equation}
In this way, we succeeded in showing that 
the condition $U\in SU(2)$ implies 
the conditions (\ref{eq:5-3-1}), (\ref{eq:5-3-2}), 
(\ref{eq:5-3-3}), (\ref{eq:5-3-4}), 
and the extra condition (\ref{eq:5-3-7}). 

We here show that adding the condition (\ref{eq:5-3-7}) 
to the conditions (\ref{eq:5-3-1}), (\ref{eq:5-3-2}), 
(\ref{eq:5-3-3}), and (\ref{eq:5-3-4}) 
completes a necessary and sufficient condition 
so that $U\in SU(2)$. 
In the case where $u_{1}u_{3}\ne 0$, 
we have the equations, $\mathrm{det}\, U=u_{1}u_{4}-u_{2}u_{3}
=|u_{1}|^{2}e^{i(\theta_{1}+\theta_{4})}-|u_{3}|^{3}e^{i(\theta_{2}+\theta_{3})}
=|u_{1}|^{2}+|u_{3}|^{2}=1$, 
since $e^{i(\theta_{1}+\theta_{4})}=1$ and $e^{i(\theta_{2}+\theta_{3})}=\, -1$ 
by the condition (\ref{eq:5-3-7}). 
In the case $u_{1}u_{3}=0$, 
if $u_{3}=0$ (resp. $u_{1}=0$), then we have the equality, 
$|u_{1}|=1$ (resp. $|u_{3}|=1$ ) by Eq.(\ref{eq:5-3-1}). 
Thus, we reach the computation, 
$\mathrm{det}\, U=|u_{1}|^{2}e^{i(\theta_{1}+\theta_{4})}=1$ 
(resp. $\mathrm{det}\, U=\, -|u_{3}|^{2}e^{i(\theta_{2}+\theta_{3})}=1$ ) 
by the condition (\ref{eq:5-3-7}).      
Therefore, the condition, $U\in SU(2)$, is equivalent to 
the conditions, (\ref{eq:5-3-1}), (\ref{eq:5-3-2}), 
(\ref{eq:5-3-3}), (\ref{eq:5-3-4}), 
and (\ref{eq:5-3-7}): 
\begin{equation}
U\in SU(2) 
\Longleftrightarrow 
\textrm{(\ref{eq:5-3-1}), (\ref{eq:5-3-2}), 
(\ref{eq:5-3-3}), (\ref{eq:5-3-4}), 
and (\ref{eq:5-3-7})}. 
\label{eq:5-3-8}
\end{equation}     

Based on this equivalence (\ref{eq:5-3-8}), set our desired 
complex numbers as 
$\alpha:=u_{1}$ and $\beta:=u_{3}$, respectively. 
Then, we have 
$u_{2}=0=\, -u_{3}^{*}$ if $u_{3}=0$, 
and $u_{2}=|u_{3}|e^{i\theta_{2}}=\, -|u_{3}|e^{-i\theta_{3}}
=\, -u_{3}^{*}$ if $u_{3}\ne 0$, 
by (\ref{eq:5-3-2}) and (\ref{eq:5-3-7}). 
We have 
$u_{4}=0= u_{1}^{*}$ if $u_{1}=0$, and 
$u_{4}=|u_{1}|e^{i\theta_{4}}=|u_{1}|e^{-i\theta_{1}}=u_{1}^{*}$ 
if $u_{1}\ne 0$, 
by (\ref{eq:5-3-3}) and (\ref{eq:5-3-7}). 
Thus, we obtain the statement of our lemma. 
\qed 
\end{proof}

Now we prove Proposition \ref{prop:03'}. 
We use the matrix representation (\ref{eq:representation-matrix}) again. 

In the case where $u_{3}\neq 0$, through the equivalence (\ref{eq:5-3-8'}), 
multiplying the both sides of Eq.(\ref{eq:5-3-4}) 
by $|u_{1}|$ gives us the expression 
$|u_{1}|^{2}(e^{i(\theta_{1}+\theta_{4})}+
e^{i(\theta_{2}+\theta_{3})})=0$. 
Thus, we can compute the determinant of $U$ 
as
\begin{eqnarray*}
\mathrm{det}\, U=u_{1}u_{4}-u_{2}u_{3}
&=&
|u_{1}|^{2}e^{i(\theta_{1}+\theta_{4})}-
|u_{3}|^{2}e^{i(\theta_{2} + \theta_{3})} \\ 
&=&\, -(|u_{1}|^{2}+|u_{3}|^{2})e^{i(\theta_{2} + \theta_{3})} 
= e^{i(\theta_{2} + \theta_{3}+\pi)}. 
\end{eqnarray*} 
Here we used Eq.(\ref{eq:5-3-1}) of the equivalence 
(\ref{eq:5-3-8'}), and the equality $e^{i\pi}=\, -1$. 
Thus, we realize that $e^{-i(\theta_{2}+\theta_{3}+ \pi)/2}U\in SU(2)$. 
Define our desired complex numbers as 
$\gamma_{1}:= e^{-i(\theta_{2}+\theta_{3}+ \pi)/2}u_{1}= 
e^{-i(\theta_{2}+\theta_{3}+ \pi)/2}\alpha$, 
$\gamma_{2}:= e^{-i(\theta_{2}+\theta_{3}+\pi)/2}u_{3}
= e^{-i(\theta_{2}+\theta_{3}+\pi)/2}\beta$, 
and $\gamma_{3}:=e^{i(\theta_{2}+\theta_{3}+ \pi)/2}$, respectively. 
Then, Lemma \ref{lem:03} gives the representation: 
$$
U=e^{i(\theta_{2} + \theta_{3}+\pi)/2}
e^{-i(\theta_{2} + \theta_{3}+\pi)/2}U
=\gamma_{3}
\left(
\begin{array}{cc}
\gamma_{1} & -\gamma_{2}^{*} \\ 
\gamma_{2} & \gamma_{1}^{*}
\end{array}
\right). 
$$

On the other hand, in the case $u_{3}= 0$, 
we can compute the determinant of $U$ 
as $\mathrm{det}\, U=e^{i(\theta_{1}+\theta_{4})}$ 
since we have the value of $|u_{1}|^{2}$ as $|u_{1}|^{2}=1$ 
by Eq.(\ref{eq:5-3-1}) of the equivalence 
(\ref{eq:5-3-8'}). 
Thus, we realize that $e^{-i(\theta_{1}+\theta_{4})/2}U\in SU(2)$. 
Based on Lemma \ref{lem:03}, 
define our desired complex numbers as 
$\gamma_{1}:= e^{-i(\theta_{1}+\theta_{4})/2}u_{1}=e^{-i(\theta_{1}+\theta_{4})/2}\alpha$, 
$\gamma_{2}:= e^{-i(\theta_{1}+\theta_{4})/2}u_{3}=e^{-i(\theta_{1}+\theta_{4})/2}\beta$, 
and $\gamma_{3}:=e^{i(\theta_{1}+\theta_{4})/2}$, 
respectively. 
Then, we reach the conclusion,  
$$
U=e^{i(\theta_{1} + \theta_{4})/2}
e^{-i(\theta_{1} + \theta_{4})/2}U
=\gamma_{3}
\left(
\begin{array}{cc}
\gamma_{1} & -\gamma_{2}^{*} \\ 
\gamma_{2} & \gamma_{1}^{*}
\end{array}
\right). 
$$

These are the construction of the representation 
in our proposition.

\subsection{Proof of Theorem \ref{theo:2}} 
\label{subsec:proof-theo-2} 

Before proving Theorem \ref{theo:2}, we make a small remark: 
For every $\psi\in D(H_{U})$, there are $\psi_{0}\in D(H_{0})$ 
and $c_{L}, c_{R}\in\mathbb{C}$ so that 
\begin{equation}
\psi=\psi_{0}+c_{L}\psi_{L}^{+}+c_{R}\psi_{R}^{+}+U(c_{L}\psi_{L}^{+}+c_{R}\psi_{R}^{+})
\label{eq:HU}
\end{equation} 
by Propositions \ref{prop:von-Neumann} and 
\ref{prop:deficiency-indices} together with 
Eq.(\ref{eq:solution}).

We prove Theorem \ref{theo:2} i) here. 
Let us suppose that $U$ is diagonal. 
In this case, it is clear that there are complex numbers 
$\gamma_{L}, \gamma_{R}\in\mathbb{C}$ 
so that 
$$
U=
\left(
\begin{array}{cc}
\gamma_{L} & 0 \\ 
0 & \gamma_{R}
\end{array}
\right),\quad 
|\gamma_{L}|=|\gamma_{R}|=1, 
$$
and thus, the operation of the unitary operator $U$ on 
$\mathcal{K}_{+}(H_{0})$ is determined by  
$U\psi_{L}^{+}=\gamma_{L}\psi_{L}^{-}$ 
and 
$U\psi_{R}^{+}=\gamma_{R}\psi_{R}^{-}$. 
By Eq.(\ref{eq:HU}), 
we can represent the boundary value $\psi(-\Lambda)$ as
$$
\psi(-\Lambda)=c_{L}\psi_{L}^{+}(-\Lambda)
+c_{L}\gamma_{L}\psi_{L}^{-}(-\Lambda)  
=c_{L}Ne^{-\sqrt{1+m^{2}}\, \Lambda}
\left(
\begin{array}{c}
1+\gamma_{L} \\ 
-\mu+\gamma_{L}\mu^{*}
\end{array}
\right),
$$
and the boundary value $\psi(+\Lambda)$ as
$$
\psi(+\Lambda)
=c_{R}\psi_{R}^{+}(+\Lambda)+c_{R}\gamma_{R}\psi_{R}^{-}(+\Lambda)
=c_{R}Ne^{-\sqrt{1+m^{2}}\, \Lambda}
\left(
\begin{array}{c}
1+\gamma_{R} \\ 
\mu-\gamma_{R}\mu^{*}
\end{array}
\right).
$$
We set $\theta_{\mu}$ as $\theta_{\mu}:=\arg\mu$, and so, 
we have $\mu=e^{i\theta_{\mu}}$. 
Here $\mu$ was given as $\mu=(1+im)/\sqrt{1+m^{2}}$, 
and thus, $\cos\theta_{\mu}=1/\sqrt{1+m^{2}}$ 
and $\sin\theta_{\mu}=m/\sqrt{1+m^{2}}$.  
We compare the boundary values 
$\psi_{\uparrow}(-\Lambda)$ and $\psi_{\downarrow}(-\Lambda)$: 
 
In the case where $\gamma_{L} \ne\, -1$, we have
\begin{eqnarray*}
\frac{\psi_{\downarrow}(-\Lambda)}{\psi_{\uparrow}(-\Lambda)}
&{=}&\frac{-\mu+\gamma_{L}\mu^{*}}{1+\gamma_{L}}  \\ 
&{=}& \frac{(-\mu+\gamma_{L}\mu^{*})(1+\gamma_{L}^{*})}{(1+\gamma_{L})
(1+\gamma_{L}^{*})}  
=\, -\, \frac{\mu-\gamma_{L}\mu^{*}
+\gamma_{L}^{*}\mu-\mu^{*}}{2+\gamma_{L}+\gamma_{L}^{*}} \\ 
&{=}& \, -i\frac{\sin(\theta_{\mu}-\theta_{L})
+\sin\theta_{\mu}}{1+\cos\theta_{L}}  
=i\left(\cos\theta_{\mu}\tan\frac{\theta_{L}}{2}-\sin\theta_{\mu}\right). 
\end{eqnarray*}
The value of $\cos\theta_{\mu}\tan(\theta_{L}/2)-\sin\theta_{\mu}$ 
runs over the whole $\mathbb{R}$ when the angular $\theta_{L}$ 
runs over $[0, 2\pi) \setminus \{ \pi \}$, 
and then, the correspondence 
$[0, 2\pi) \setminus \{ \pi \}\ni\theta_{L}\longrightarrow 
\rho_{-}\in\mathbb{R}$ 
makes the one-to-one correspondence.  
On the other hand, in the case where $\gamma_{L} =\, -1$, 
we have 
$\psi_{\uparrow}(-\Lambda)=0$ and 
$\psi_{\downarrow}(-\Lambda)=\, -c_{L}Ne^{-\sqrt{1 + m^2} \Lambda}(\mu+\mu^{*})$.

Similarly, compare the boundary values 
$\psi_{\uparrow}(+\Lambda)$ and $\psi_{\downarrow}(+\Lambda)$: 
In the case where $\gamma_{R}\ne\, -1$, we have 
$$
\frac{\psi_{\downarrow}(+\Lambda)}{\psi_{\uparrow}(+\Lambda)} 
=\frac{\mu-\gamma_{R}\mu^{*}}{1+\gamma_{R}} 
=i\left(-\cos\theta_{\mu}\tan\frac{\theta_{R}}{2}
+\sin\theta_{\mu}\right).
$$
The correspondence 
$[0, 2\pi) \setminus \{ \pi \}\ni\theta_{R}\longrightarrow 
\rho_{+}\in\mathbb{R}$ 
makes the one-to-one correspondence. 
In the case where $\gamma_{R}=\, -1$, we have 
$\psi_{\uparrow}(+\Lambda)=0$ and 
$\psi_{\downarrow}(+\Lambda)=c_{R}Ne^{-\sqrt{1+ m^{2}} \Lambda}(\mu+\mu^{*})$. 

Therefore, we realize that the condition, $D(H_{\rho})=D(H_{U})$, 
is equivalent to the correspondence: 
$i\rho_{-}=
\psi_{\downarrow}(-\Lambda)/\psi_{\uparrow}(-\Lambda)$ 
for $\gamma_{L}\ne\, -1$ 
and $\rho_{-}
=|\psi_{\downarrow}(-\Lambda)/\psi_{\uparrow}(-\Lambda)|=+\infty$ 
for $\gamma_{L}=\, -1$, 
and 
$i\rho_{+}=
\psi_{\downarrow}(+\Lambda)/\psi_{\uparrow}(+\Lambda)$ 
for $\gamma_{R}\ne\, -1$ 
and 
$\rho_{+}=|\psi_{\downarrow}(+\Lambda)/\psi_{\uparrow}(+\Lambda)|=+\infty$ 
for $\gamma_{R}=\, -1$, 
which gives our desired correspondence.
 
We prove Theorem \ref{theo:2} ii) now. 
First up, 
Proposition \ref{prop:03'} gives the representation of $U$: 
there are complex numbers $\gamma_{1}, \gamma_{2}, \gamma_{3}\in\mathbb{C}$ 
so that   
$$
U=\gamma_{3}
\left(
\begin{array}{cc}
\gamma_{1} & -\gamma_{2}^{*} \\ 
\gamma_{2} & \gamma_{1}^{*}
\end{array}
\right)\,\,\, 
\textrm{with}\,\,\,  
|\gamma_{1}|^{2}+|\gamma_{2}|^{2}
=|\gamma_{3}|= 1,\, 
\gamma_{2} \neq 0. 
$$
Here the fact, $\gamma_{2}\ne 0$, comes from the assumption that 
$U$ is non-diagonal. 
Thus, the operation of $U$ on $\mathcal{K}_{+}(H_{0})$ 
is determined by  
$U\psi_{L}^{+}=\gamma_{1}\gamma_{3}\psi_{L}^{-}
-\gamma_{2}^{*}\gamma_{3}\psi_{R}^{-}$ 
and 
$U\psi_{R}^{+}=\gamma_{2}\gamma_{3}\psi_{L}^{-}
+\gamma_{1}^{*}\gamma_{3}\psi_{R}^{-}$. 
Using Eq.(\ref{eq:HU}), we can compute individual boundary values 
$\psi(-\Lambda)$ and $\psi(+\Lambda)$ as 
\begin{eqnarray*}
\psi(-\Lambda)
&{=}& c_{L}\psi_{L}^{+}(-\Lambda)+c_{L}\gamma_{1}\gamma_{3}\psi_{L}^{-}(-\Lambda)
+c_{R}\gamma_{2}\gamma_{3}\psi_{L}^{-}(-\Lambda) \\ 
&{=}& Ne^{-\sqrt{1+m^{2}}\Lambda}
\left(
\begin{array}{cc}
1+\gamma_{1}\gamma_{3} & \gamma_{2}\gamma_{3} \\ 
-\mu+\gamma_{1}\gamma_{3}\mu^{*} & \gamma_{2}\gamma_{3}\mu^{*}
\end{array}
\right)
\left(
\begin{array}{c}
c_{L} \\ 
c_{R}
\end{array}
\right)
\end{eqnarray*}
and 
\begin{eqnarray*}
\psi(+\Lambda)
&{=}& c_{R}\psi_{R}^{+}(+\Lambda)
-c_{L}\gamma_{2}^{*}\gamma_{3}\psi_{R}^{-}(+\Lambda)
+c_{R}\gamma_{1}^{*}\gamma_{3}\psi_{R}^{-}(+\Lambda) \\ 
&{=}& Ne^{-\sqrt{1+ m^{2}} \Lambda}
\left(
\begin{array}{cc}
-\gamma_{2}^{*}\gamma_{3} & 1+\gamma_{1}^{*}\gamma_{3} \\ 
\gamma_{2}^{*}\gamma_{3}\mu^{*} 
& \mu-\gamma_{1}^{*}\gamma_{3}\mu^{*}
\end{array}
\right)
\left(
\begin{array}{c}
c_{L} \\ 
c_{R}
\end{array}
\right). 
\end{eqnarray*}
We remember that 
$\gamma_{2}\neq 0$ and $\gamma_{3}\neq 0$, 
and then, 
$$
\mathrm{det}\, 
\left(
\begin{array}{cc}
1+\gamma_{1}\gamma_{3} & \gamma_{2}\gamma_{3} \\ 
-\mu+\gamma_{1}\gamma_{3}\mu^{*}  
& \gamma_{2}\gamma_{3}\mu^{*}
\end{array}
\right)
=\gamma_{2}\gamma_{3}(\mu + \mu^{*})
= \frac{2\gamma_{2}\gamma_{3}}{\sqrt{1+m^{2}}}\ne 0.  
$$
Thus, noting  $\gamma_{3}^{-1}=\gamma_{3}^{*}$, 
we can compute the following inverse matrix: 
$$
\left(
\begin{array}{cc}
1+\gamma_{1}\gamma_{3} & \gamma_{2}\gamma_{3} \\ 
-\mu+\gamma_{1}\gamma_{3}\mu^{*} 
& \gamma_{2}\gamma_{3}\mu^{*}
\end{array}
\right)^{-1}
=\frac{1}{\gamma_{2}(\mu + \mu^{*})}
\left(
\begin{array}{cc}
\gamma_{2}\mu^{*} & -\gamma_{2} \\ 
\gamma_{3}^{*}\mu-\gamma_{1}\mu^{*} 
& \gamma_{1}+\gamma_{3}^{*}
\end{array}
\right).
$$
Thus, define a $2\times 2$ matrix $V$ by 
$$
V\equiv
\left(
\begin{array}{cc}
v_{1} & v_{2} \\ 
v_{3} & v_{4} \\
\end{array}
\right)  
:= 
\left(
\begin{array}{cc}
-\gamma_{2}^{*}\gamma_{3} 
& 1+\gamma_{1}^{*}\gamma_{3} \\ 
\gamma_{2}^{*}\gamma_{3}\mu^{*} 
& \mu-\gamma_{1}^{*}\gamma_{3}\mu^{*}
\end{array}
\right)
\left(
\begin{array}{cc}
1+\gamma_{1}\gamma_{3} & \gamma_{2}\gamma_{3} \\ 
-\mu+\gamma_{1}\gamma_{3}\mu^{*} 
& \gamma_{2}\gamma_{3}\mu^{*}
\end{array}
\right)^{-1}.   
$$
Then, we have 
$$
V=
\frac{\sqrt{1+m^{2}}}{\gamma_{2}}
\left(
\begin{array}{cc}
i\left\{\Im(\gamma_{1}^{*}\mu)+\Im(\gamma_{3}^{*}\mu)\right\} 
& \Re\gamma_{1}+\Re\gamma_{3} \\ 
-\Re\gamma_{1}+\Re(\gamma_{3}^{*}\mu^{2}) 
& i\left\{\Im(\gamma_{1}\mu)+\Im(\gamma_{3}^{*}\mu)\right\}
\end{array}
\right). 
$$
Thus, we reach the boundary condition: 
$\psi(+\Lambda)=V\psi(-\Lambda)$ 
for every $\psi \in D(H_{U})$. 
We set $v_{j}'$ as 
$v_{j}':=i(\gamma_{2}/|\gamma_{2}|)v_{j}$, $j = 1, \cdots , 4$, 
and then, we have $v_{j}'{v_{k}'}^{*}=v_{j}v_{k}^{*}$. 
Then, $v_{1}'$ and $v_{4}'$ are real numbers, 
and $v_{2}'$ and $v_{3}'$ are purely imaginary numbers, 
which implies the relations: 
$\Re(v_{1}v_{2}^{*})=\Re(v_{1}'{v_{2}'}^{*})=0$, 
$\Re(v_{1}v_{3}^{*})=\Re(v_{1}'{v_{3}'}^{*})=0$, 
$\Re(v_{2}v_{4}^{*})=\Re(v_{2}'{v_{4}'}^{*})=0$, 
and $\Re(v_{3}v_{4}^{*})=\Re(v_{3}'{v_{4}'}^{*})=0$. 
So, we have confirmed the first part of conditions 
of the class (\ref{Cl-alpha}). 
We check the last two conditions of the class (\ref{Cl-alpha}): 
The immediate computation easily bring us 
to $v_{1}v_4^{*}+v_{2}v_{3}^{*}=1$ using
$|\gamma_{1}|^{2}+|\gamma_{2}|^{2}=1$.  
We here note that $v_{k}^{*}=v_{k}(\gamma_{2}/\gamma_{2}^{*})$, $k=2, 3$, 
which implies $v_{2}v_{3}^{*}=v_{2}\{ v_{3}(\gamma_{2}/\gamma_{2}^{*})\}
=\{ v_{2}(\gamma_{2}/\gamma_{2}^{*})\}v_{3}=v_{2}^{*}v_{3}^{*}$. 
Thus, we have $v_{1}v_{4}^{*}+v_{2}^{*}v_{3}=v_{1}v_{4}^{*}+v_{2}v_{3}^{*}=1$. 
Therefore, we can conclude from the above argument 
that the vector $v=(v_{1}, v_{2}, v_{3},v_{4}) \in \mathbb{C}^{4}$ 
is in the class (\ref{Cl-alpha}), 
and then, $V=B_{v}$. 
Therefore, the condition, $D(H_{\alpha})=D(H_{U})$, 
is equivalent to the correspondence 
$\alpha=v$. 
We accomplished the proof of the part ii). 

\subsection{Proof of Proposition \ref{prop:06}} 

We denote by $\mathcal{A}_{0}$ 
the set on the right hand side of our desired representation. 
It is evident that $\mathcal{A}_{0}\subset\mathcal{A}$. 
So, the only thing we have to do is that 
we show $\mathcal{A}\subset\mathcal{A}_{0}$. 
For every $B_{\alpha}\in\mathcal{A}$,  
set $\theta_{j}$ as $\theta_{j}= \arg\alpha_{j}$. 
Since the vector $\alpha$ is in the class (\ref{Cl-alpha}),
$\alpha_{1}\alpha_{2}^{*}$, $\alpha_{1}\alpha_{3}^{*}$, 
$\alpha_{2}\alpha_{4}^{*}$, and $\alpha_{3}\alpha_{4}^{*}$ are purely 
imaginary numbers. 
Moreover, the last condition of the class (\ref{Cl-alpha}) says 
that $\alpha_{2}\alpha_{3}^{*}=1-\alpha_{1}\alpha_{4}^{*}=\alpha_{2}^{*}\alpha_{3}$. 
That is, $\alpha_{2}\alpha_{3}^{*}$ is a real number. 
Thus, it follows from the last condition, 
that $\alpha_{1}\alpha_{4}^{*}$ is also a real number, 
and $\alpha_{1}\ne 0$ or $\alpha_{3}\ne 0$. 
In the case where $\alpha_{1}\neq 0$, 
setting $\theta\in[0,2\pi)$, and 
$a_{1}, a_{2}, a_{3}, a_{4}\in\mathbb{R}$ 
as $\theta:=\arg(\alpha_{1}/|\alpha_{1}|)$, 
$a_{1}:=|\alpha_{1}|$, 
$a_{2}:=(\alpha_{1}\alpha_{2}^{*})^{*}/i|\alpha_{1}|$, 
$a_{3}:=(\alpha_{1}\alpha_{3}^{*})^{*}/i|\alpha_{1}|$, 
and $a_{4}:=(\alpha_{1}\alpha_{4}^{*})^{*}/|\alpha_{1}|$, 
we immediately obtain the representation of $B_{\alpha}$ 
in $\mathcal{A}_{0}$. 
In the case $\alpha_{1}=0$, 
we only have to set $\theta\in[0,2\pi)$, and 
$a_{1}, a_{2}, a_{3}, a_{4}\in\mathbb{R}$ 
by $\theta:=\arg(-i\alpha_{3}/|\alpha_{3}|)$, 
$a_{1}:=i\alpha_{1}\alpha_{3}^{*}/|\alpha_{3}|$, 
$a_{2}:=\alpha_{2}\alpha_{3}^{*}/|\alpha_{3}|$, 
$a_{3}:=|\alpha_{3}|$, 
and $a_{4}:=i(\alpha_{3}\alpha_{4}^{*})^{*}/|\alpha_{3}|$, 
respectively, 
and then, 
we reach our desired fact $B_{\alpha}\in\mathcal{A}_{0}$. 
Thus, the two cases imply that 
$\mathcal{A}\subset\mathcal{A}_{0}$.  
Therefore, we can conclude the proof of 
the equality, $\mathcal{A}=\mathcal{A}_{0}$. 

\subsection{Proof of Proposition \ref{prop:inverse}}
\label{subsection:prop-inverse}

We prove Proposition \ref{prop:inverse} here. 
First up, it immediately follows from the definition of $\gamma_{1}$, 
$\gamma_{2}$, and $\Gamma_{0}$ that 
$|\gamma_{1}|^{2}+|\gamma_{2}|^{2}=1$. 
We here remark that this equation gives us the equation, 
\begin{eqnarray}
1=
\Gamma_{0}^{2}
&\Bigl[ 
\frac{4}{1+m^{2}}
+\sum_{j=1}^{4}|\alpha_{j}|^{2} 
-2\Re(\mu^{*}\alpha_{1}\alpha_{2}^{*})
+2\Re(\mu^{*}\alpha_{1}\alpha_{3}^{*})
-2\Re(\mu^{*2}\alpha_{1}\alpha_{4}^{*}) 
\nonumber \\ 
&\quad 
-2\Re(\alpha_{2}\alpha_{3}^{*})
+2\Re(\mu\alpha_{2}^{*}\alpha_{4})
-2\Re(\mu\alpha_{3}^{*}\alpha_{4})
\Bigr].
\label{eq:inv-1}
\end{eqnarray}
Next, we show $|\gamma_{3}|=1$. 
It is easy to check the equations, 
\begin{eqnarray*} 
&{}& \Re(\mu\alpha_{1}\alpha_{j}^{*})
=\, -\Re(\mu^{*}\alpha_{1}\alpha_{j}^{*})
=\frac{m}{\sqrt{1+m^{2}}}a_{1}a_{j},\quad 
j=2, 3, \\ 
&{}& \Re(\mu^{*}\alpha_{j}^{*}\alpha_{4})
=\, -\Re(\mu\alpha_{j}^{*}\alpha_{4})
=\, -\, \frac{m}{\sqrt{1+m^{2}}}a_{j}a_{4},\quad 
j=2, 3, \\ 
&{}& \Re(\mu^{*2}\alpha_{j}\alpha_{k}^{*})
= \frac{1-m^{2}}{1+m^{2}} \Re(\alpha_{j}\alpha_{k}^{*})
= \frac{1-m^{2}}{1+m^{2}} a_{j}a_{k},\quad 
(j,k)=(1,4), (2,3), 
\end{eqnarray*}
by Proposition \ref{prop:06}. 
By using these equations together with 
$$
\frac{1\mp m^{2}}{1\pm m^{2}}
= \frac{2}{1\pm m^{2}}\, -1,
$$ 
we have 
\begin{eqnarray*} 
|\gamma_{3}|^{2}&{=}&
\Gamma_{0}^{2}
\Bigl[ 
\sum_{j=1}^{4}|\alpha_{j}|^{2} 
+2\Re(\mu\alpha_{1}\alpha_{2}^{*})
+2\Re(\mu^{*}\alpha_{1}\alpha_{3}^{*})
+2\Re(\alpha_{1}\alpha_{4}^{*})  
+2\Re(\mu^{*2}\alpha_{2}\alpha_{3}^{*}) \\ 
&{}&\qquad 
+2\Re(\mu\alpha_{2}^{*}\alpha_{4})
+2\Re(\mu^{*}\alpha_{3}^{*}\alpha_{4})
\Bigr] \\ 
&{=}& 
\textrm{right hand side of (\ref{eq:inv-1})}.
\end{eqnarray*}
Thus, we have $|\gamma_{3}|=1$, 
and then, we reach 
$$
U=
\gamma_{3}
\left(
\begin{array}{cc}
\gamma_{1} & -\gamma_{2}^{*} \\ 
\gamma_{2} & \gamma_{1}^{*}
\end{array}
\right)
\in U(1)S\mathbb{H}=U(2).
$$

Thus, what we have to show is that every 
$\psi\in D(H_{U})$ satisfies the boundary condition (\ref{BC-alpha}). 
Insert the boundary values $\psi_{\sharp}(+\Lambda)$ and 
$\psi_{\sharp}(-\Lambda)$ with expressions, 
$$
\left\{ \begin{array}{l}
\psi_{\sharp}(+\Lambda)
=c_{R}\psi_{R\sharp}^{+}(+\Lambda)
-c_{L}\gamma_{3}\gamma_{2}^{*}\psi_{R\sharp}^{-}(+\Lambda)
+c_{R}\gamma_{3}\gamma_{1}^{*}\psi_{R\sharp}^{-}(+\Lambda), \\  
\psi_{\sharp}(-\Lambda)
=c_{L}\psi_{L\sharp}^{+}(-\Lambda)
+c_{L}\gamma_{3}\gamma_{1}\psi_{L\sharp}^{-}(-\Lambda)
+c_{R}\gamma_{3}\gamma_{2}\psi_{L\sharp}^{-}(-\Lambda), 
\end{array}\right.
\sharp=\uparrow, \downarrow,
$$ 
into the boundary conditions, 
$$
\left\{ \begin{array}{l}
\psi_{\uparrow}(+\Lambda)
=\alpha_{1}\psi_{\uparrow}(-\Lambda)
+\alpha_{2}\psi_{\downarrow}(-\Lambda), \\ 
\psi_{\downarrow}(+\Lambda)
=\alpha_{3}\psi_{\uparrow}(-\Lambda)
+\alpha_{4}\psi_{\downarrow}(-\Lambda).  
\end{array}\right.
$$
Then, by using the arbitrariness of the coefficients $c_{L}$ 
and $c_{R}$ in $D(H_{U})$ 
and noting the fact $\gamma_{3}^{-1}=\gamma_{3}^{*}$,  
we can show that the condition $D(H_{U})=D(H_{\alpha})$ 
is equivalent to the system of the following system of equations: 
\begin{eqnarray}
&{}& (\alpha_{1}+\mu^{*}\alpha_{2})\gamma_{1}+\gamma_{2}^{*}
=\gamma_{3}^{*}(-\alpha_{1}+\mu\alpha_{2}),
\label{eq:inverse-2} \\  
&{}& (\alpha_{1}+\mu^{*}\alpha_{2})\gamma_{2}-\gamma_{1}^{*}
=\gamma_{3}^{*}, 
\label{eq:inverse-3} \\  
&{}& (\alpha_{3}+\mu^{*}\alpha_{4})\gamma_{1}-\mu^{*}\gamma_{2}^{*}
=\gamma_{3}^{*}(-\alpha_{3}+\mu\alpha_{4}), 
\label{eq:inverse-4} \\  
&{}& (\alpha_{3}+\mu^{*}\alpha_{4})\gamma_{2}+\mu^{*}\gamma_{1}^{*}
=\mu\gamma_{3}^{*}.
\label{eq:inverse-5}
\end{eqnarray}
Then, we can show that our $(\gamma_{1}, \gamma_{2}, \gamma_{3})$ 
is a solution of this system of equations: 

Noting $\mu+\mu^{*}=2/\sqrt{1+m^{2}}$ and $a_{1}a_{4}+a_{2}a_{3}=1$, 
we have $\mu a_{1}a_{4}+\mu^{*}a_{2}a_{3}=
2/\sqrt{1+m^{2}}-\mu^{*}a_{1}a_{4}-\mu a_{2}a_{3}$. 
Thus, we realize that our $\gamma_{1}, \gamma_{2}, \gamma_{3}$ 
satisfy (\ref{eq:inverse-2}) as 
\begin{eqnarray*}
(\alpha_{1}+\mu^{*}\alpha_{2})\gamma_{1}+\gamma_{2}^{*}
&{=}& 
ie^{i\theta}\Gamma_{0}
\Bigl[
-\mu^{*}a_{1}^{2}+i(1-\mu^{*2})a_{1}a_{2}
-ia_{1}a_{3}-\mu^{*}a_{1}a_{4} \\ 
&{}&\qquad\qquad  
-\mu^{*}a_{2}^{2}-\mu a_{2}a_{3}+ia_{2}a_{4}
\Bigr] \\ 
&{=}&
\gamma_{3}^{*}(-\alpha_{1}+\mu\alpha_{2}).
\end{eqnarray*}
Using $\mu\mu^{*}=1$ and $\mu^{*2}=2\mu^{*}/\sqrt{1+m^{2}}-1$, 
we can show that our $\gamma_{1}, \gamma_{2}, \gamma_{3}$ 
satisfy (\ref{eq:inverse-3}) as 
$$
(\alpha_{1}+\mu^{*}\alpha_{2})\gamma_{2}-\gamma_{1}^{*}
= \Gamma_{0}
\left(
i\mu^{*}a_{1}-\mu^{*2}a_{2}-a_{3}+i\mu^{*}a_{4}
\right)  
= \gamma_{3}^{*}.
$$
Combining $\mu^{*2}=2\mu^{*}/\sqrt{1+m^{2}}-1$ and $a_{1}a_{4}+a_{2}a_{3}=1$, 
we have $-\mu^{*2}a_{1}a_{4}-a_{2}a_{3}=
a_{1}a_{4}^{*}+\mu^{*2}a_{2}a_{3}-2\mu^{*}/\sqrt{1+m^{2}}$. 
Using this equation and $\mu\mu^{*}=1$, 
we can show that our $\gamma_{1}, \gamma_{2}, \gamma_{3}$ 
satisfy (\ref{eq:inverse-4}) as 
\begin{eqnarray*}
(\alpha_{3}+\mu^{*}\alpha_{4})\gamma_{1}-\mu^{*}\gamma_{2}^{*}
&{=}& i\Gamma_{0}e^{i\theta}
\Bigl[
-i\mu^{*}a_{1}a_{3}+a_{1}a_{4}+\mu^{*2}a_{2}a_{3}+i\mu^{*}a_{2}a_{4} \\ 
&{}&\qquad\qquad 
+a_{3}^{2}+i(\mu-\mu^{*})a_{3}a_{4}+a_{4}^{2}
\Bigr] \\
&{=}& \gamma_{3}^{*}(-\alpha_{3}+\mu\alpha_{4}). 
\end{eqnarray*}
Using $\mu\mu^{*}=1$ and $\mu+\mu^{*}=2/\sqrt{1+m^{2}}$, 
we know that  our $\gamma_{1}, \gamma_{2}, \gamma_{3}$ 
satisfy (\ref{eq:inverse-5}) as 
\begin{eqnarray*}
(\alpha_{3}+\mu^{*}\alpha_{4})\gamma_{2}+\mu^{*}\gamma_{1}^{*}
=\Gamma_{0}
\left(
ia_{1}-\mu^{*}a_{2}-\mu a_{3}+ia_{4}
\right)
=\mu\gamma_{3}^{*}.
\end{eqnarray*}
 
Therefore, consequently, we can complete the proof of our proposition. 

\section{Conclusion}
\label{sec:Conclusion}

We have proved that all the boundary conditions of wave functions of our Dirac particle 
are completely classified into the two types. 
For the case where the electron's wave functions do not pass through the junction, 
their boundary condition can be described by two parameters, $\gamma_{L}, \gamma_{R}\in\mathbb{C}$ 
with $|\gamma_{L}|=|\gamma_{R}|=1$, determined by von Neumann's theory. 
In the case where the wave functions do pass through the junction, 
the boundary condition is described by Benvegn\`{u} and D\c{a}browski's four-parameter family, 
and then, their four parameters can actually be described 
by three parameters, $\gamma_{1}, \gamma_{2}, \gamma_{3} 
\in \mathbb{C}$ with $|\gamma_{1}|^{2}+|\gamma_{2}|^{2}=|\gamma_{3}|=1$ and $\gamma_{2}\ne 0$, 
determined by von Neumann's theory. 
These results stem from our one-to-one correspondence formulae, 
Eqs.(\ref{eq:TJF1}) and (\ref{eq:TJF2}) 
with Propositions \ref{prop:03'} and \ref{prop:06}.

Let us make small two remarks at the tail end of this paper. 
Using our method, we can completely classify the boundary conditions 
of all self-adjoint extensions of the minimal Schr\"{o}dinger 
operator, too \cite{HK13-S}. 
In the Dirac operator's case, there is no effect of 
the length of junction in the boundary condition. 
However, in the Schr\"{o}dinger operator's case\textcolor{red}{,} 
we can find it in the boundary condition. 
We have not understand any strictly physical reason 
why the Schr\"{o}dinger particle feels the length 
$2\Lambda$ of the junction, but the Dirac particle does not. 
We conjecture that the speed of the particle is concerned 
with the reason.

\section*{Acknowledgments}
One of the authors (M.H.) acknowledges the financial support from JSPS, 
Grant-in-Aid for Scientific Research (C) 23540204. 
He also expresses special thanks to 
Kae Nemoto and Yutaka Shikano for the useful discussions 
with them.


\begin{thebibliography}{99}

\bibitem{AGH-KH} 
\textsc{S. Albeverio, F. Gesztesy, R. H{\o}gh-Krohn, and H. Holden}, 
\textit{Solvable models in quantum mechanics}, Springer, New York, 1988. 



\bibitem{ADeV00}
\textsc{V. Alonso and S. De Vincenzo}, 
\textit{Delta-type Dirac point interactions and their nonrelativistic limits}, 
Int. J. Math. Phys., \textbf{39} (2000), 1483. 

\bibitem{A64}
\textsc{A. F. Andreev}, 
\textit{Thermal conductivity of the intermediate state 
of superconductors.}, 
Sov. Phys. JETP, \textbf{19} (1964), 1228.

\bibitem{A65}
\textsc{A. F. Andreev}, 
\textit{Thermal conductivity of the intermediate state 
of superconductors. II}, 
Sov. Phys. JETP, \textbf{20} (1965), 1490.



\bibitem{AFS02}
\textsc{D. D. Awschalom, M. E. Flatt\'{e}, and N. Samarth},  
\textit{Spintronics}, 
Scientific American, \textbf{286} (2002), 67. 

\bibitem{BMN08} 
\textsc{J. Behrndt, M. Malamud, and H. Neidhardt}, 
\textit{Scattering matrices and Weyl functions}, 
Proc. London Math. Soc., \textbf{97} (2008), 568. 

\bibitem{BNRRW08} 
\textsc{J. Behrndt, H. Neidhardt, E. R. Racec, P. N. Racec, U. Wulf}, 
\textit{On Eisenbud's and Wigner's R-matrix: A general approach}, 
J. Differ. Equations, \textbf{244} (2008), 2545. 

\bibitem{BD94}
\textsc{S. Benvegn\`{u} and L. D\c{a}browski}, 
\textit{Relativistic point interaction}, 
Lett. Math. Phys., \textbf{30} (1994), 159.

\bibitem{BGMBPA06}
\textsc{J. Berezovsky, O. Gywat, F. Meier, 
D. Battaglia, X. Peng, and D. D. Awschalom}, 
\textit{Initialization and read-out of spins in
coupled core-shell quantum dots}, 
Nature Physics, \textbf{2} (2006), 831. 


\bibitem{BMN02} 
\textsc{J. F. Brasche, M. Malamud, and H. Neidhardt}, 
\textit{Weyl function and spectral properties 
of self-adjoint extensions}, 
Integr. Equ. Oper. Theo., \textbf{43} (2002), 264. 


\bibitem{BGP08}
\textsc{J. Br\"{u}ning, V. Geyler, and K. Pankrashkin}, 
\textit{Spectra of self-adjoint extensions and applications 
to solvable Schr\"{o}dinger operators}, 
Rev. Math. Phys., \textbf{20} (2008), 1. 


\bibitem{Burkard06}
\textsc{G. Burkard},  
\textit{Spin qubits: Connect the dots}, 
Nature Physics, \textbf{2} (2006), 807.  

\bibitem{CMP13}
\textsc{R. Carlone, M. Malamud, and A. Posilicano}, 
\textit{On the spectral theory of Gesztesy-\v{S}eba realizations of 
$1$-$D$ Dirac operators with point interactions on a discrete set}, 
arXiv:1302.5044v1. 

\bibitem{D-AM89} 
\textsc{F. Dominguez-Adame and E. Marcia}, 
\textit{Bound states and confining properties of rel. point interaction
potentials}, 
J. Phys. A, \textbf{22} (1989), L419.

\bibitem{FHNS10} 
\textsc{Y. Furuhashi, M. Hirokawa, K. Nakahara, and Y. Shikano}, 
\textit{ Role of a phase factor in the boundary condition of a one-dimensional junction}, 
J. Phys. A: Math. Theo., \textbf{43} (2010), 354010. 

\bibitem{GBPSSPFRS10} 
\textsc{A. B. Giroday, A. J. Bennett, M. A. Pooley, 
R. M. Stevenson, N. Sk\"{o}ld, R. B. Patel, 
I. Farrer, D. A. Ritchie, and A. J. Shields}, 
\textit{All-electrical coherent control of the exciton states in a single quantum dot}, 
Phys. Rev. B, \textbf{82} (2010), 241301.  

\bibitem{UniTokyo11}
\textsc{S. Hermelin, S. Takada, M. Yamamoto, 
S. Tarucha, A. D. Wieck, L. Saminadayar, 
C. B\"{a}uerle, and T. Meunier}, 
\textit{Electrons surfing on a sound wave as a platform for quantum optics with flying electrons},  
Nature, \textbf{477} (2011), 435. 

\bibitem{Hir00} 
\textsc{M. Hirokawa}, 
\textit{Canonical quantization on a doubly connected space and the Aharonov-Bohm phase}, 
J. Funct. Anal., \textbf{174} (2000), 322. 

\bibitem{HK13-S}
\textsc{M. Hirokawa and T. Kosaka}, 
\textit{One-dimensional tunnel-junction formula for Schr\"{o}dinger particle}, 
arXiv:0759041v1.  

\bibitem{Hughes97} 
\textsc{R. J. Hughes}, 
\textit{Relativistic point interactions: Approximation by smooth potentials}, 
Rep. Math. Phys., \textbf{39} (1997), 425.

\bibitem{LLPA02}
\textsc{M. N. Leuenberger, D. Loss, M. Poggio, and D. D. Awschalom}, 
\textit{Quantum information processing with large nuclear spins in GaAs semiconductors}, 
Phys. Rev. Lett., \textbf{89} (2002), 207601.

\bibitem{LFA05}
\textsc{M. N. Leuenberger, M. E. Flatt\'{e}, and D. D. Awschalom}, 
\textit{Teleportation of electronic many-qubit states encoded 
in the electron spin of quantum dots via single photons}, 
Phys. Rev. Lett., \textbf{94} (2005), 107401. 

\bibitem{LDV98} 
\textsc{D. Loss and D. P. DiVincenzo}, 
\textit{Quantum computation with quantum dots},  
Phys. Rev. A, \textbf{57} (1998), 120.  

\bibitem{Cambridge11}
\textit{R. P. G. McNeil, M. Kataoka, C. J. B. Ford, 
C. H. W. Barnes, D. Anderson, 
G. A. C. Jones, I. Farrer, and D. A. Ritchie}, 
\textit{On-demand single-electron transfer between distant quantum dots}, 
Nature, \textbf{477} (2011), 439. 

\bibitem{Tonomura2}
\textsc{N. Osakabe, T. Matsuda, T. Kawasaki, J. Endo, A. Tonomura, 
S. Yano, and H. Yamada}, 
\textit{Experimental confirmation of Aharonov-Bohm effect using a toroidal magnetic field confined by a superconductor}, 
Phys. Rev. A, \textbf{34} (1986), 815.

\bibitem{PT13}
\textsc{S. Pedersen and F. Tian}, 
\textit{Momentum operators in the unit square}, 
Integr. Equ. Oper. Theo., \textbf{77} (2013), 57. 


\bibitem{RS2}
\textsc{M. Reed and B. Simon},  
 \textit{Methods of Modern Mathematical Physics II. 
Fourier Analysis, Self-Adjointness}, 
Academic Press, San Diego, 1975. 


\bibitem{MIT07}
\textsc{T. S. Santos, J. S. Lee, P. Migdal, 
I. C. Lekshmi, B. Satpati, and J. S. Moodera}, 
\textit{Room-temperature tunnel magnetoresistance and spin-polarized tunneling through an organic semiconductor barrier},  
Phys. Rev. Lett., \textbf{98} (2007), 016601. 

\bibitem{Eindhoven09}
\textsc{J. J. H. M. Schoonus, P. G. E. Lumens, W. Wagemans, 
J. T. Kohlhepp, P. A. Bobbert, H. J. M. Swagten, and B. Koopmans}, 
\textit{Magnetoresistance in hybrid organic spin valves at the onset of multiple-step tunneling}, 
Phys. Rev. Lett., \textbf{103} (2009), 146601.  

\bibitem{Seba89}
\textsc{P. \v{S}eba}, \textit{Klein's paradox and the relativistic point interaction}, 
Lett. Math. Phys., \textbf{18} (1989), 77.


\bibitem{SH11} 
\textsc{Y. Shikano and M. Hirokawa}, 
\textit{Boundary conditions in one-dimensional tunneling junction}, 
J. Phys.: Conference Series, \textbf{302} (2011), 012044.

\bibitem{TOD08} 
\textsc{A. Tokuno, M. Oshikawa, and E. Demler}, 
\textit{Dynamics of one-dimensional Bose liquids: 
Andreev-like reflection at Y-junctions and the absence of the 
Aharonov-Bohm effect}, 
Phys. Rev. Lett., \textbf{100} (2008), 140402. 

\bibitem{Tonomura1}
\textsc{A. Tonomura, N. Osakabe, T. Matsuda, T. Kawasaki, and J. Endo}, 
\textit{Evidence for Aharonov-Bohm effect with magnetic field completely shielded from electron wave}, 
Phys. Rev. Lett., \textbf{56} (1986), 792.

\bibitem{weidmann}
\textsc{J. Weidmann}, 
\textit{Linear Operators in Hilbert Spaces}, 
Springer-Verlag, New York, 1980.  

\bibitem{UniTokyo12}
\textsc{M. Yamamoto, S. Takada, C. B\"{a}uerle, K. Watanabe, A. D. Wieck, and S. Tarucha}, 
\textit{Electrical control of a solid-state flying qubit}, 
Nature Nanotechnology, \textbf{7} (2012), 247.


\bibitem{Tohoku13}
\textsc{X. Zhang, S. Mizukami, T. Kubota, Q. Ma, M. Oogane, 
H. Naganuma, Y. Ando, and T. Miyazaki}, 
\textit{Observation of a large spin-dependent transport length in organic spin valves at room temperature},  
Nature Communications, \textbf{4} (2013), 1392.



\end{thebibliography}
\end{document}